\documentclass[11pt,a4paper]{amsart}

\usepackage[left=1in,right=1in,top=1in,bottom=1in,foot=0.5in]{geometry}

\usepackage{amsmath,amsfonts,amssymb}
\usepackage[T1]{fontenc}
\usepackage{color}
\usepackage{epsfig}
\usepackage{graphicx}
\usepackage{cite,url}
\usepackage{hyperref}
\usepackage{cleveref}
\usepackage{xspace}
\usepackage{enumerate}

\usepackage{tikz}
\usepackage{multicol}
\usepackage[linesnumbered]{algorithm2e}

\usepackage{caption}
\usepackage{subcaption}

\newtheorem{thml}{Theorem}

\newtheorem{mytheorem}{Theorem}[section]
\newtheorem{mylemma}[mytheorem]{Lemma}
\newtheorem{myproposition}[mytheorem]{Proposition}
\newtheorem{mycorollary}[mytheorem]{Corollary}
\newtheorem{myclaim}[mytheorem]{Claim}

\theoremstyle{definition}

\newcommand{\cO}{\ensuremath{\mathcal{O}}\xspace}

\DeclareMathOperator{\ecc}{ecc}

\DeclareMathOperator{\proj}{pr}

\DeclareMathOperator{\diam}{diam}
\DeclareMathOperator{\rad}{rad}

\begin{document}

\title{A fine-grained dichotomy for the center problem on Gromov hyperbolic graphs}

\author[G.\ Ducoffe]{Guillaume Ducoffe} \address{National Institute
 for Research and Development in Informatics and University of
  Bucharest, Bucureşti, Rom\^{a}nia} \email{guillaume.ducoffe@ici.ro}

\date{}

\begin{abstract}
A vertex in a graph is called {\em central} if it minimizes its maximum distance to the other vertices. The {\em radius} of a graph $G$ is the largest distance between a central vertex and the other vertices, and it is denoted by $\rad(G)$. In the center problem, we are asked to find a central vertex.
 We study the fine-grained complexity of the center problem on graphs with small Gromov hyperbolicity. Roughly, the Gromov hyperbolicity of a graph represents how close, locally, it is to a tree, from a metric point of view. It has applications in the design of approximation algorithms. In particular, there is a linear-time algorithm that for every $\delta$-hyperbolic graph $G$ outputs some vertex at distance at most $\rad(G) + 5\delta$ to the other vertices [Chepoi et al, {\it SoCG}'08]. However, a linear-time algorithm for computing a central vertex is known only for $0$-hyperbolic graphs, whereas its existence was ruled out for $2$-hyperbolic graphs under the Hitting Set Conjecture of [Abboud et al, {\it SODA}'16]. Our main contribution in the paper is a linear-time algorithm for computing a central vertex in the class of $\frac 1 2$-hyperbolic graphs. Furthermore, we rule out the existence of such an algorithm for $1$-hyperbolic graphs, under the Hitting Set Conjecture, thus completely settling all the cases left open. 
%
\end{abstract}

\maketitle

\section{Introduction}\label{sec:intro}
We study the finer-grained complexity of facility location problems on graphs with respect to some properties of their distance function. 
This is in contrast with structural parameterizations for these problems, such as treewidth~\cite{AVW16} and clique-width~\cite{jDuc22e}.
The structure of graphs sharing distance properties with the classical metric and geometric spaces (\textit{e.g.}, Euclidean spaces, $\ell_p$-spaces, Hyperbolic spaces, etc.) has long been studied~\cite{BaCh-survey}. These prior works are mostly devoted to graphs that are undirected, unweighted, connected, and simple (\textit{i.e.}, with no loops nor multiple edges). Such is also the setting considered in the paper. We refer to~\cite{BSCW+-weighted-hamming,Cab-convex} for other related works on edge-weighted undirected graphs. For standard graph terminology, see also~\cite{BoMu08}. 
Algorithmic applications of some of these properties to distance computation problems can be found in~\cite{BCC+22,VCDim,DrDuGu_Helly_hyp}. 
We focus on the \textsc{Center} problem, the definition of which is outlined in what follows.

\medskip
\noindent
{\bf The center problem.}
Let $G$ be a graph. For every two vertices $u$ and $v$, their distance, denoted by $d_G(u,v)$, equals the minimum number of edges on a $uv$-path. Let $\ecc_G(v) = \max\{ d_G(u,v) : u \in V(G) \}$ be the \emph{eccentricity} of vertex $v$. The \emph{radius} and the \emph{diameter} of $G$ are defined as $\rad(G) = \min\{ \ecc_G(v) : v \in V(G) \}$ and $\diam(G) = \max\{ \ecc_G(v) : v \in V(G) \}$, respectively. We will omit the subscript if the graph $G$ is clear from the context. A vertex is called \emph{central} if its eccentricity equals the radius. The \textsc{Center} problem asks to find a central vertex. It is a fundamental problem in Network Analysis, and in Location Theory. Note that a naive algorithm that first computes the distance matrix allows one to solve the \textsc{Center} problem in $\cO(nm)$ time on $n$-vertex $m$-edge graphs. This runtime is essentially optimal assuming the Hitting Set Conjecture~\cite{AVW16}, the definition of which is recalled in Sec.~\ref{sec:hard}.
In particular, for the design of linear-time algorithms, we are required to consider more restricted classes of graphs.
Our focus in the paper is on graphs with small Gromov hyperbolicity.

\medskip
\noindent
{\bf Gromov hyperbolic graphs.}
The \emph{Gromov hyperbolicity} of a connected graph $G$ is the least value $\delta$ such that for every four vertices $u,v,x,y$, it holds that $d(u,v) + d(x,y) \le \max\{d(u,x)+d(v,y), d(u,y)+d(v,x)\} + 2\delta$~\cite{Gr}. The definition applies to general metric spaces. It is a relaxation of the four-point characterization of tree metrics~\cite{Bun-trees}.  Note that for graphs, the Gromov hyperbolicity is either a natural number or a positive half-integer (of the form $k/2$, where $k$ is a natural odd number). It has been experimentally verified that the Gromov hyperbolicity is small in practice on some biological networks and communication networks~\cite{AlDr,CCL}, therefore suggesting the use of this parameter in order to better classify complex networks~\cite{AbDr16,KSN16}, and to explain some of their properties~\cite{BoChCa15,Chepoi:2017:CCI:3039686.3039835,JoLo04}. Furthermore, unlike many graph width parameters that are NP-hard to compute, or even to approximate, the Gromov hyperbolicity of a graph can be computed in polynomial time. The naive $\cO(n^4)$-time algorithm has been improved in~\cite{FIV15}. There exist parameterized algorithms~\cite{FKMN+19}, as well as approximation algorithms~\cite{HypApprox,Duan14}, that are even faster. See also~\cite{CNV22}, and the references therein for practical algorithms.

On the algorithmic side, every $\delta$-hyperbolic graph can be embedded in a tree with additive distortion in $\cO(\delta \log{n})$~\cite{Gr}. The latter result can be used in the design of compact labeling schemes for approximate distance computation~\cite{DLS}. For other algorithmic applications, see~\cite{HypCuts,PTAS-TSP}. Our starting point for this work is that for every $\delta$-hyperbolic graph $G$, an almost central vertex: of eccentricity at most $\rad(G)+5\delta$, can be computed in $\cO(n+m)$ time~\cite{ChDrEsHaVa}. This seminal result was later generalized to the approximate computation of all the vertex eccentricities~\cite{Chepoi2018FastEA}, and to an approximation algorithm for the more general $k$-{\sc Center} problem~\cite{ChEs,EKS}. In the paper, we consider \underline{exact} algorithms for the \textsc{Center} problem. 

The combination of Gromov hyperbolicity with other structural and metric properties can be used in the design of linear-time algorithms for the \textsc{Center} problem on various graph classes~\cite{Gpunimodal-ecc,DrDuGu_Helly_hyp}. However, Gromov hyperbolicity on its own is not a strong enough parameter for subquadratic-time center computation; in particular, assuming the Hitting Set conjecture, the radius of $2$-hyperbolic graphs cannot be computed in truly subquadratic-time~\cite{Chepoi2018FastEA}. Insofar, the only positive result in this direction is the linear-time algorithm for computing a central vertex of a $0$-hyperbolic graph. More specifically, the $0$-hyperbolic graphs are exactly the block graphs~\cite{How}, a proper subclass of chordal graphs, and a linear-time algorithm for the \textsc{Center} problem on chordal graphs is known~\cite{ChDr}. (Recall that a graph is chordal if it has no induced cycle of length more than three.) Therefore, the complexity of the \textsc{Center} problem has been left open for the following hyperbolicity values: $\frac 1 2, 1 \ \text{and} \ \frac 3 2$. We stress that for some well-studied classes of graphs, including chordal graphs, weakly chordal graphs, AT-free graphs, distance-hereditary graphs, cocomparability graphs, the link graphs of simple polygons, and the graphs of $7$-systolic simplicial complexes, the hyperbolicity is at most one~\cite{hypChordal,ChDrEsHaVa,WuZh01}. Therefore, bridging the complexity gap for the {\sc Center} problem between $0$-hyperbolic graphs and $2$-hyperbolic graphs is a legitimate research direction. A similar question could be asked for other distance-related problems, such as diameter computation. In this work, we completely settle all of the cases left open for the {\sc Center} problem.

 \subsection{Contributions}\label{ssec:contributions}
 Our main result in the paper is as follows:

 \begin{thml}\label{thm:center:half-hyp}
    The \textsc{Center} problem can be solved in linear time on $\frac 1 2$-hyperbolic graphs.
 \end{thml}

 Before our work, the best-known algorithms for the {\sc Center} problem on (superclasses of) $\frac 1 2$-hyperbolic graphs were the deterministic $\cO(m^{1.71})$-time algorithm from~\cite{alpha-center}, and the recent randomized $\Tilde{\cO}(m^{\frac 3 2})$-time algorithm from~\cite{Gpunimodal-ecc}. These running-times were outmatching the naive $\cO(nm)$-time algorithm only for graphs of sufficiently low density. Furthermore, unlike~\cite{Gpunimodal-ecc}, all our algorithms in the paper are deterministic.

\smallskip
A full characterization of the $\frac 1 2$-hyperbolic graphs was given in~\cite{BaCh2003}.
Their structure is arguably more complex than that of block graphs.
In particular, we stress that every $C_4$-free graph is an induced subgraph of some $\frac 1 2$-hyperbolic graph~\cite{CoDu14}.
Furthermore, every $C_4$-free graph within the following classes of graphs is $\frac 1 2$-hyperbolic: AT-free graphs, cocomparability graphs, permutation graphs and distance-hereditary graphs~\cite{WuZh01}.
In contrast to Theorem~\ref{thm:center:half-hyp}, let us mention that the best-known algorithm for the \textsc{Center} problem on AT-free graphs and on cocomparability graphs runs in $\cO(m^{3/2})$ time~\cite{Du_ATfree}.

\medskip
\noindent
{\bf Overview of our approach.}
In order to sketch our strategy to prove Theorem~\ref{thm:center:half-hyp}, we need to introduce a few more notations and terminology.
For every two vertices $u$ and $v$, let the (metric) \emph{interval} $I(u,v)$ be defined as the set of all vertices $w$ on a shortest $uv$-path (\textit{i.e.}, such that $d(u,v) = d(u,w) + d(w,v)$).
For every $k$ such that $0 \le k \le d(u,v)$, let $S_k(u,v) = \{ w \in I(u,v) : d(u,w) = k\}$ be called a \emph{slice}.
Our basic strategy would be to compute a pair $u,v$ of mutually distant vertices (\textit{i.e.}, such that $\ecc(u) = \ecc(v) = d(u,v)$), using a few BFS, then to extract a central vertex from a middle slice.
This approach has been used in the design of linear-time algorithms for the \textsc{Center} problem on various graph classes~\cite{ChDr_HHD,DuDr,Ola}.
A refinement of it was presented in~\cite{ChDr} for chordal graphs: such that either we extract a central vertex from a middle slice \textit{or} we compute a new pair $u',v'$ of mutually distant vertices such that $d(u',v') > d(u,v)$.
Since for chordal graphs $d(u,v) \ge \diam(G)-2$, a central vertex is found after a constant number of iterations.
Our algorithm follows a similar strategy to these prior works, although we must tackle with new challenges that are specific to $\frac 1 2$-hyperbolic graphs.
More specifically, in some cases where $d(u,v) = 2r-1$ is odd, previous techniques can neither be used to extract a central vertex from the middle slices $S_{r-1}(u,v)$ and $S_r(u,v)$, nor to find a new pair $u',v'$ of distance larger than $d(u,v)$.
In this situation, we extend the search for a central vertex to the neighborhoods of both slices. 
The latter requires inspecting the vertices at a distance of up to $3$ from the slices in order to guide the search.
Hence, as a side contribution of this work, we bring new insights on the structure of the balls centered at a clique of a $\frac 1 2$-hyperbolic graph.

Another difference between~\cite{ChDr} and our Theorem~\ref{thm:center:half-hyp} is the implementation of the distance-to-clique procedure.
Roughly, the procedure must output some distance information for every vertex of a clique, which includes their respective eccentricities.
The linear-time implementation proposed in~\cite[Sec. 3]{ChDr} was based on a property of chordal graphs that does \emph{not} hold for $\frac 1 2$-hyperbolic graphs (namely, that cliques in a chordal graph are outergated, see Sec.~\ref{sec:basics}).
We propose a different implementation, which is based on recent results about cliques in a graph with convex balls (a superclass of chordal graphs and $\frac 1 2$-hyperbolic graphs)~\cite{Gpunimodal-ecc}. 

\medskip
\noindent
{\bf Additional results.}
Recall that the algorithm presented in this paper is proved to be correct under the unchecked assumption that the input graph is $\frac 1 2$-hyperbolic. 
If we run this algorithm on an arbitrary input, then it may fail, or it may output a non-central vertex.
The latter raises the question of deciding whether a graph is $\frac 1 2$-hyperbolic.
The problem was studied in~\cite{CoDu14}, where a subcubic equivalence was proved between the recognition of $\frac 1 2$-hyperbolic graphs and the detection of an induced $C_4$.
Combined with the results from~\cite{WWWY}, it implies that we can decide whether a graph is $\frac 1 2$-hyperbolic in $\cO(n^{\omega+o(1)})$ time, where $\omega < 2.371339$ denotes the square matrix multiplication exponent.
We complete these prior works with the following conditional lower bound (the definition of the Strong Exponential-Time Hypothesis is recalled in Sec.~\ref{sec:hard}):

\begin{thml}\label{thm:hardness:half-hyp}
    Under the Strong Exponential-Time Hypothesis, the recognition of $\frac 1 2$-hyperbolic graphs requires $\Omega(n^{2-o(1)})$ time, even on $n$-vertex split graphs with at most $n^{1+o(1)}$ edges.
\end{thml}

A similar hardness result was proved in~\cite{BoCrHa16}, but only for the recognition of $1$-hyperbolic graphs.
The $1$-hyperbolic graphs are much less structured than the $\frac 1 2$-hyperbolic graphs.
For example, every graph with diameter at most three is $1$-hyperbolic~\cite{KoMo}.

\smallskip
Finally, we complete Theorem~\ref{thm:center:half-hyp} with the following hardness result on the {\sc Center} problem for $1$-hyperbolic graphs:

\begin{thml}\label{thm:hardness:1-hyp}
    Assuming the Hitting Set conjecture, the \textsc{Center} problem on $1$-hyperbolic graphs requires $\Omega(n^{2-o(1)})$ time, even on graphs with $n$ vertices and at most $n^{1+o(1)}$ edges.
\end{thml}

\noindent
{\bf Perspectives: other definitions of hyperbolicity.}
There are other so-called ''negative curvature'' parameters on graphs, that can only differ from Gromov hyperbolicity by a small multiplicative factor; see~\cite{HypApprox}.
One of these parameters, the slimness, has received special attention~\cite{Slimness}.
Although the differences between these parameters and Gromov hyperbolicity are irrelevant in the design of approximation algorithms, they may lead to varying behaviors, for small values, when considering exact algorithms. For instance, we can prove that linear-time center computation can only be achieved for graphs with slimness $0$ ({\it a.k.a.}, the block graphs), whereas it essentially requires quadratic-time, under the Hitting Set conjecture, for any positive value of slimness. The latter result is a byproduct of our proof for Theorem~\ref{thm:hardness:1-hyp}.

\section{Preliminaries}\label{sec:basics}

In what follows, we introduce the notations, terminology, and prior results that are required for the presentation of the main algorithm and its analysis.
Basic concepts of neighborhoods, distances, eccentricities, and other related notions are both defined for vertices and vertex-subsets.
This will make easier the presentation of our algorithms in Sec.~\ref{sec:alg}.

Let $G=(V,E)$ be a graph.
The neighborhood of a vertex $x$, denoted by $N(x)$, is the set of all the vertices $y$ such that $xy \in E$.
The degree of $x$ is equal to $d(x) = |N(x)|$.
For a vertex-set $X$, we define similarly $N(X) = \{ y \in V \setminus X : N(y) \cap X \ne \emptyset \}$, and $d(X) = |N(X)|$.
We recall that the distance $d(x,y)$ between two vertices $x$ and $y$ is the minimum number of edges on a path between $x$ and $y$.
An $xy$-path with minimum number of edges is called a shortest $xy$-path.
The ball of center $x$ and radius $r$, hereafter denoted by $B_r(x)$, is the set of all the vertices $y$ such that $d(x,y) \le r$.
Let $\ecc(x) = \max\{ d(x,y) : y \in V \}$ be the eccentricity of $G$.
The set of vertices that are furthest from $x$ is denoted by $F(x) = \{ y \in V : d(x,y) = \ecc(x)\}$.
Two vertices $x$ and $y$ are mutually distant if $x \in F(y)$ and $y \in F(x)$.
Let $\rad(G) = \min\{\ecc(x) : x \in V\}$ and $\diam(G) = \max\{\ecc(x) : x \in V\}$ be called the radius and the diameter of $G$, respectively.
A vertex $x$ is called central if $\ecc(x) = \rad(G)$.

These notations can be extended from vertices $x$ to vertex-sets $X$ as follows.
For every vertex $y$ let $d(y,X) = \min\{ d(x,y) : x \in X \}$.
For every natural number $r$ let $B_r(X) = \{ y \in V : d(y,X) \le r \}$.
Let $\ecc(X) = \max\{ d(y,X) : y \in V \}$, and let $F(X) = \{ y \in V : d(y,X) = \ecc(X)\}$.
The (weak) diameter of $X$ is defined as $\diam(X) = \max\{ d(x,x') : x,x' \in X \}$.

Recall that for every vertices $x$ and $y$, $I(x,y) = \{z \in V : d(x,y) = d(x,z) + d(z,y)\}$.
Let also $I^o(x,y) = I(x,y) \setminus \{x,y\}$.
For every $k$ such that $0 \le k \le d(x,y)$, let $S_k(x,y) = \{z \in I(x,y) : d(x,z) = k\}$ be called a slice.
A set $X$ is called \emph{convex} if $I(x,x') \subseteq X$ for every $x,x' \in X$.
A \emph{metric triangle} consists of three different vertices $x',y',z'$ such that the sets $I^o(x',y')$, $I^o(y',z')$, and $I^o(z',x')$ are pairwise disjoint. 
The type of a metric triangle $x',y',z'$ is defined as the triple $(d(x',y'),d(y',z'),d(z',x'))$.
A \emph{median} of a triple of vertices $x,y,z$ is any vertex of $I(x,y) \cap I(y,z) \cap I(z,x)$.
A \emph{quasi-median} of $x,y,z$ is a metric triangle $x',y',z'$ such that the following distance equalities hold: $d(x,y) = d(x,x')+d(x',y')+d(y',y)$, $d(y,z) = d(y,y')+d(y',z')+d(z',z)$, and $d(z,x) = d(z,z')+d(z',x')+d(x',x)$. 
It was observed in~\cite{ChDrEsHaVa} that every triple of vertices $x,y,z$ has a median or a quasi-median.

\medskip
\noindent
{\bf Projections, shadows, and potentials.}
The following concepts are a cornerstone of our approach in the next section.
Let $X$ be a vertex set of $G$.
The \emph{projection} of a vertex $y$ on $X$ is defined as $\proj_y(X) = \{ x \in X : d(y,x) = d(y,X) \}$.
For every vertex-set $Y$, we further define the $Y$-\emph{shadow} of a vertex $x \in X$ as being $\Psi_X(x,Y) = \{ y \in Y : x \in \proj_y(X)\}$;
let $\psi_X(x,Y) = |\Psi_X(x,Y)|$.
Of particular interest is the special case where $Y = F(X)$; then, we simply write $\Psi_X(x)$ and $\psi_X(x)$.
Finally, for every natural number $r$, the $r$-{\em potential} of a vertex $x$ with respect to $Y$ is defined as $\Phi_r(x,Y) = B_r(x) \cap Y$; let $\phi_r(x,Y) = |\Phi_r(x,Y)|$.

\medskip
\noindent
{\bf Graph classes considered in our results.}
We recall that a graph $G=(V,E)$ is $\delta$-hyperbolic if for every $u,v,x,y \in V$, $d(u,v)+d(x,y) \le \max\{d(u,x)+d(v,y),d(u,y)+d(v,x)\}+2\delta$.
In what follows, we mostly consider the special case $\delta = \frac 1 2$.
We further consider two related metric properties on graphs.
The graph $G$ satisfies the {\em $\alpha_1$-metric property} if for every $u,v,w,x \in V$ such that $v \in I(u,w)$ and $w \in I(v,x)$ are adjacent, $d(u,x) \ge d(u,v)+d(w,x)$.
Then, we also call $G$ an $\alpha_1$-metric graph.
A \emph{CB-graph} is a graph with convex balls.
It was observed in~\cite{BaCh2003} that every $\frac 1 2$-hyperbolic graph is $\alpha_1$-metric, and in~\cite{YuCh1991} that every $\alpha_1$-metric graph is a CB-graph.
We will often use these relations in our proofs with no further mention.

\smallskip
The following properties of $\delta$-hyperbolic graphs are used:
\begin{mylemma}\label{lem:hyp-pties}
    If $G=(V,E)$ is a $\delta$-hyperbolic graph, then
    \begin{enumerate}
        \item\label{item:leanness} the diameter of every slice is at most $2\delta$;
        in particular, every slice of a $\frac 1 2$-hyperbolic graph is a clique~\cite{wordprocessing};
        \item\label{item:hyp:furthest-vertex} for an arbitrary vertex $x$, for any $y \in F(x)$, $\ecc(y) \ge \diam(G)-2\delta$~\cite{ChDrEsHaVa};
        \item\label{item:mutually-distant_vertices} a pair $u,v$ of mutually distant vertices can be computed in $\cO(\delta m)$ time~\cite{Chepoi2018FastEA}.
    \end{enumerate}
\end{mylemma}

A distance-preserving subgraph in a graph $G$ is called an {\em isometric subgraph}.
A graph is called a {\em Helly graph} if every family of pairwise intersecting balls has a nonempty common intersection.
For every graph $G$ there is an inclusionwise minimal Helly graph $\mathcal{H}(G)$ in which $G$ isometrically embeds~\cite{Isbell}.
The graph $\mathcal{H}(G)$ is sometimes called the {\em injective hull} of $G$.

\begin{mylemma}[\cite{injectivehull}]\label{lem:injective-hull}
    A graph $G$ is $\delta$-hyperbolic if and only if its injective hull $\mathcal{H}(G)$ is $\delta$-hyperbolic.    
\end{mylemma}

The following results on the $\alpha_1$-metric property are also used.
\begin{mylemma}\label{lem:alpha1-pties}
    If $G=(V,E)$ is an $\alpha_1$-metric graph, then
    \begin{enumerate}
        \item\label{item:alpha1:metric-triangle} every metric triangle is of type $(1, 1, 1)$, $(1, 2, 2)$, $(2, 1, 2)$, $(2, 2, 1)$ or $(2, 2, 2)$~\cite{BaCh2003};
        \item\label{item:triangle-pty} for every adjacent $x,y \in S_k(u,v)$, $x$ and $y$ have common neighbours in both $S_{k-1}(u,v)$ and $S_{k+1}(u,v)$~\cite{alpha-hyperb}.  
    \end{enumerate}
\end{mylemma}

\medskip
\noindent
{\bf $k$-Outergated sets.}\label{ssec:k-outergated} 
Let $G=(V,E)$ be a graph.
We call $X \subseteq V$ a \emph{$k$-outergated} set if for every vertex $z$ such that $d(z,X) > k$, $\bigcap\{ S_k(x,z) : x \in \proj_z(X) \} \ne \emptyset$.
Furthermore, every vertex $z^*$ in this intersection is called a $k$-outergate of $z$ with respect to $X$.
For $k=1$, we simply refer to outergated sets, and to outergates.

\begin{mylemma}[\cite{Du_Helly}]\label{lem:compute-outergates}
    For a set $X$ of a graph $G$, in $\cO(m)$ time one can map every $z \notin X$ to $z^* \in B_{d(z,X)-1}(z) \cap B_1(X)$ maximizing $|N(z^*) \cap X|$. 
    Furthermore, if $X$ is outergated, then $z^*$ is an outergate of $z$.            
\end{mylemma}

Finally, the following recent results on cliques in a CB-graph are used in our analysis:
\begin{mylemma}[\cite{Gpunimodal-ecc}]\label{lem:cb-graph:cliques}
    If $K$ is a clique of a CB-graph $G=(V,E)$, then
    \begin{enumerate}
        \item\label{item:cb:cliques-2outergated} $K$ is $2$-outergated. More precisely, for any $z \in V \setminus K$, either $z$ has an outergate or $\proj_z(K) = K$ and $z$ has a $2$-outergate;
        \item\label{item:cb:cliques-vertex-classification} in $\cO(m)$ time we can compute the set of vertices without an outergate with respect to $K$. 
    \end{enumerate}
\end{mylemma}

\section{An algorithm for the center problem}\label{sec:alg}

The following section is devoted to the proof of Theorem~\ref{thm:center:half-hyp}.
For that, we analyze Algorithm~\ref{alg:center}.

    \subsection{Intermediate computations: shadows, potentials and eccentricities}\label{ssec:ecc-cliques}

Shadows were implicitly used in~\cite{ChDr}, where a linear-time distance-to-clique procedure is introduced in order to compute the shadows for the vertices of a clique of a chordal graph.
We present a novel version of the distance-to-clique procedure from~\cite{ChDr}, which is tailored for the $\frac 1 2$-hyperbolic graphs. 
Indeed, the initial procedure from~\cite{ChDr} was based on the assumption that cliques are outergated, which is true for chordal graphs but not for the $\frac 1 2$-hyperbolic graphs.
For example, the cycle $C_5$ is $\frac 1 2$-hyperbolic, however the edges of $C_5$ are not outergated.

\begin{mylemma}\label{lem:intermediate-computations}
    Let $K,M$ be disjoint vertex-subsets in a $\frac 1 2$-hyperbolic graph $G$ such that $K$ is a clique.
    The following values, for every $w \in K$, can be computed in $\cO(n+m)$ time:
    \begin{enumerate}
        \item\label{item:compute-shadows} $\psi_K(w,M)$ ({\it shadows});
        \item\label{item:compute-potentials} $\phi_r(w,M)$, for a fixed natural number $r$ ({\it potentials});
        \item\label{item:compute-ecc-clique} and $\ecc(w)$ ({\it eccentricities}).
    \end{enumerate}
\end{mylemma}

       \begin{proof}
            We give separate proofs for (\ref{item:compute-shadows}), (\ref{item:compute-potentials}), and (\ref{item:compute-ecc-clique}).

            \smallskip
            (\ref{item:compute-shadows}): We bi-partition the vertices of $M$ according to the existence of an outergate in $K$. In particular, let $M_1$ contain every vertex of $M$ that has an outergate with respect to $K$, and let $M_2 = M \setminus M_1$.
            By Lemma~\ref{lem:cb-graph:cliques}.\ref{item:cb:cliques-vertex-classification}, this partitioning can be done in $\cO(n+m)$ time.
            For the vertices with an outergate, we adapt the distance-to-clique procedure from~\cite{ChDr} for shadow computations.
            More specifically, we compute outergates $g_K(z)$, for every $z \in M_1$, which can be done in $\cO(n+m)$ time by Lemma~\ref{lem:compute-outergates}.
            By doing so, we obtain for every $s \in N(K)$ the value $\gamma(s) = \left|\{z \in M_1 : s = g_K(z)\}\right|$.
            Furthermore, by Lemma~\ref{lem:cb-graph:cliques}.\ref{item:cb:cliques-2outergated}, $M_2 \subseteq \Psi_K(w,M)$ for every $w \in K$.
            Then, for every $w \in K$, $\psi_K(w,M) = |M_2| + \sum\{ \gamma(s) : s \in N(K) \cap N(w) \}$, which can be computed in $\cO(d(w))$ time.
            The overall running-time is in $\cO(n+m)$.

            In what follows, the computation of $r$-potentials (\ref{item:compute-potentials}), and of eccentricities (\ref{item:compute-ecc-clique}) is reduced to that of some shadows (\ref{item:compute-shadows}) for different subsets $M'$.

            \smallskip
            (\ref{item:compute-potentials}): We execute a BFS with start subset $K$.
        By doing so, we compute the following subsets: $M_{< r} = \{ z \in M : d(z,M) < r\}$; and $M_r = \{z \in M : d(z,M) = r\}$.
        Since $K$ is a clique, for every $w \in K$, $\phi_r(w,M) = |M_{< r}| + \psi_K(w,M_r)$.
        In particular, computing all values $\phi_r(w,M)$ can be reduced in $\cO(n+m)$ time to computing all values $\psi_K(w,M_r)$, which can be done in $\cO(n+m)$ time by (\ref{item:compute-shadows}).

            \smallskip
            (\ref{item:compute-ecc-clique}): We may assume that $K$ and $F(K)$ are disjoint: else, all vertices of $G$ are in $K$, and so, their eccentricities are equal to $1$.
        Apply (\ref{item:compute-shadows}) for $M=F(K)$.
        For every $w \in K$, $\ecc(w) = \ecc(K)$ if $\psi_K(w) = |F(K)|$, otherwise $\ecc(w) = \ecc(K)+1$ because $K$ is a clique.
        \end{proof}

Note that Lemma~\ref{lem:intermediate-computations}.\ref{item:compute-ecc-clique} is also a special case of~\cite[Prop. 11.10]{Gpunimodal-ecc}.

\subsection{Lower and upper bounds on eccentricities}\label{ssec:ecc-bounds}

We next provide upper bounds on the eccentricities of the vertices in a middle slice. 

    \begin{mylemma}\label{lem:ecc-ub}
        If $u,v$ are mutually distant vertices in a $\frac 1 2$-hyperbolic graph $G=(V,E)$, then $\ecc(w) \le \left\lceil \frac{d(u,v)} 2 \right\rceil + 1$ for every $w \in S_{\left\lfloor \frac{d(u,v)} 2 \right\rfloor}(u,v)$ (resp. for every $w \in S_{\left\lceil \frac{d(u,v)} 2 \right\rceil}(u,v)$). 

        \smallskip
        In particular,
        $$\left\lceil \frac{d(u,v)} 2 \right\rceil \le \rad(G) \le \left\lceil \frac{d(u,v)} 2 \right\rceil + 1$$
    \end{mylemma}
    \begin{proof}
        By symmetry (up to reverting the respective roles of $u$ and $v$), it suffices to prove the result for the slice $S_{\left\lfloor \frac{d(u,v)} 2 \right\rfloor}(u,v)$.
        Let $z \in V$ be arbitrary.
        It suffices to prove that $d(w,z) \le \left\lceil \frac{d(u,v)} 2 \right\rceil + 1$ for every $w \in S_{\left\lfloor \frac{d(u,v)} 2 \right\rfloor}(u,v)$.
        For that, since $u$ and $v$ are mutually distant, $d(u,z) \le d(u,v)$ and $d(v,z) \le d(u,v)$.
        In particular, the balls $B_{\left\lfloor \frac{d(u,v)} 2 \right\rfloor}(u),B_{\left\lceil \frac{d(u,v)} 2 \right\rceil}(v) \ \text{and} \ B_{\left\lceil \frac{d(u,v)} 2 \right\rceil}(z)$ pairwise intersect.
        Let $\mathcal{H}(G)$ be the injective hull of $G$.
        Since $\mathcal{H}(G)$ is a Helly graph, there is a vertex $c_z$ of $\mathcal{H}(G)$ such that $d(c_z,u) = \left\lfloor \frac{d(u,v)} 2 \right\rfloor$, $d(c_z,v) = \left\lceil \frac{d(u,v)} 2 \right\rceil$, and $d(c_z,z) \le \left\lceil \frac{d(u,v)} 2 \right\rceil$.
        By Lemma~\ref{lem:injective-hull}, $\mathcal{H}(G)$ is $\frac 1 2$-hyperbolic.
        Therefore, by Lemma~\ref{lem:hyp-pties}.\ref{item:leanness}, $d(w,c_z) \le 1$ for every $w \in S_{\left\lfloor \frac{d(u,v)} 2 \right\rfloor}(u,v)$.
        It implies $d(z,w) \le d(z,c_z)+d(c_z,w) \le \left\lceil \frac{d(u,v)} 2 \right\rceil + 1$.
    \end{proof}

    Lemma~\ref{lem:ecc-ub} can be directly used in order to solve the even case of our algorithm (case of mutually distant vertices at an even distance).

    \begin{mycorollary}\label{cor:even-case}
         If $u,v$ are mutually distant vertices in a $\frac 1 2$-hyperbolic graph $G$ such that $d(u,v) = 2r$ is even, then every vertex of minimum eccentricity in $S_r(u,v)$ is central.
    \end{mycorollary}
    \begin{proof}
        If $\rad(G) = r$, then every central vertex is in $S_r(u,v)$.
        Otherwise, by Lemma~\ref{lem:ecc-ub}, $\rad(G)=r+1$, and every vertex of $S_r(u,v)$ is central.
    \end{proof}

    For the odd case, our first intent was to reuse the strategy from~\cite{ChDr}: either we extract a central vertex from a middle slice, or we find a new pair of mutually distant vertices at a larger distance.
For that, the following procedure can be used:

\begin{mylemma}\label{lem:larger-x}
    Let $K,M$ be disjoint vertex-subsets in an $\alpha_1$-metric graph $G$ such that $K$ is a clique, and for every $z \in M$, $d(z,K) \ge r$.
    If the vertex $c$ maximizes $\psi_K(c,M)$, but $\psi_K(c,M) < |M|$, then for every $x \in M \setminus \Psi_K(c,M)$, $\ecc(x) \ge \max\{d(x,y) : y \in M\} \ge 2r$.
\end{mylemma}
\begin{proof}
    Let $w \in \proj_x(K)$.
    By the maximality of $\psi_K(c,M)$, there is an $y \in M$ such that $c \in \proj_y(K)$ and $w \notin \proj_y(K)$.
    Then, by the $\alpha_1$-metric property, $\ecc(x) \ge d(x,y) \ge d(x,w)+d(c,y) \ge 2r$.
\end{proof}

\begin{mycorollary}\label{cor:larger-x}
    If $K$ is a clique in an $\alpha_1$-metric graph $G$, the vertex $c$ maximizes $\psi_K(c)$, and $\ecc(c) = \ecc(K)+1$, then for every $x \in F(c)$, $\ecc(x) \ge 2\ecc(K)$.    
\end{mycorollary}
\begin{proof}
    Apply Lemma~\ref{lem:larger-x} for $M=F(K)$ and $r=\ecc(K)$.
    Since we assume $\ecc(c) = r+1$, $\psi_K(c,M) < |M|$.
    Furthermore, since $K$ is a clique, $F(c) = M \setminus \Psi_K(c,M)$.
\end{proof}

However, the strategy sketched above can only be used in order to solve the odd case $d(u,v) = 2r-1$ if the middle slice $S_{r-1}(u,v)$ has eccentricity $r$. 
The example of Fig.~\ref{fig:odd-case} shows that it is not always the case.

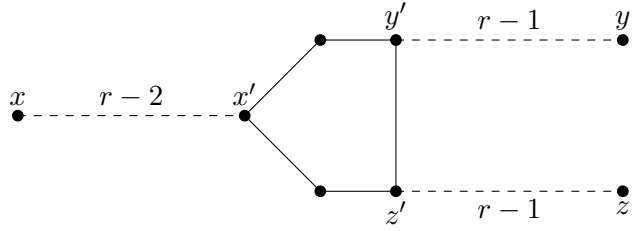
\begin{figure}[!h]
    \centering
    \begin{tikzpicture}
        \filldraw (-3,0) circle (2pt) node[anchor=south]{$x$};
        \filldraw (0,0) circle (2pt) node[anchor=south]{$x'$};
        \filldraw (1,1) circle (2pt) node[anchor=south]{};
        \filldraw (1,-1) circle (2pt) node[anchor=south]{};
        \filldraw (2,1) circle (2pt) node[anchor=south]{$y'$};
        \filldraw (2,-1) circle (2pt) node[anchor=north]{$z'$};
        \filldraw (5,1) circle (2pt) node[anchor=south]{$y$};
        \filldraw (5,-1) circle (2pt) node[anchor=north]{$z$};

        \draw[dashed] (-3,0) -- (0,0);
        \draw[dashed] (2,1) -- (5,1);
        \draw[dashed] (2,-1) -- (5,-1);
        \draw (0,0) -- (1,1) -- (2,1) -- (2,-1) -- (1,-1) -- (0,0);

        \draw node at (-1.5,.25) {$r-2$};
        \draw node at (3.5,1.25) {$r-1$};
        \draw node at (3.5,-1.25) {$r-1$};
    \end{tikzpicture}
    \caption{The graph $G$, made of three paths pending to the vertices of a $C_5$, is $\frac 1 2$-hyperbolic. Vertices $x,y,z$ are pairwise at distance $2r-1$, however $d\left(y,S_{r-1}(x,z)\right) = d\left(z,S_{r-1}(x,y)\right) = r+1$.}
    \label{fig:odd-case}
\end{figure}

\subsection{The odd case}\label{ssec:odd-case}

We present in what follows the new results and techniques that are required to solve the odd case.
For starters, we use metric triangles (Lemma~\ref{lem:alpha1-pties}.\ref{item:alpha1:metric-triangle}) to prove that we can always restrict the search for a central vertex to the two middle slices and their respective neighborhoods.
\begin{mylemma}\label{lem:odd-case:locate-center}
    Let $u,v,c$ be vertices in a graph $G$ such that $d(u,c) \le r$, $d(v,c) \le r$, and $d(u,v) = 2r-1$.
    If $G$ is $\alpha_1$-metric, then $c \in B_1\left(S_{r-1}(u,v)\right) \cup B_1\left(S_r(u,v)\right)$. 
\end{mylemma}
\begin{proof}
    If $d(u,c) = r-1$ then $c \in S_{r-1}(u,v)$, and similarly if $d(v,c) = r-1$ then $c \in S_r(u,v)$.
    Therefore from now on we assume that $d(u,c) = d(v,c) = r$.
  
    Suppose by contradiction there exists a median vertex $x$ for $u,v,c$.
    Without loss of generality, $d(u,x) \ge d(v,x)$.
    Since $d(u,v) = d(u,x)+d(x,v) = 2r-1$, $d(u,x) \ge r$.
    However, since $x \ne c$, the latter implies $d(u,c) > d(u,x) = r$, which is a contradiction.
    Therefore, in what follows let $u',v',c'$ be a quasi-median for $u,v,c$.

    Assume first $d(u',c') = d(v',c') = k$.
    By Lemma~\ref{lem:alpha1-pties}.\ref{item:alpha1:metric-triangle}, $k \le 2$ and $d(u',v') \le k$.
    In this situation, $d(u,u') = d(v,v') = r - k - d(c,c') \le r-k$.
    Furthermore, $2r-1 = d(u,v) = d(u,u') + d(u',v') + d(v',v) \le (r-k) + k + (r-k) = 2r-k$.
    Hence, $d(u',c') = d(v',c') = d(u',v') = 1$.
    It implies $d(u,u') = d(v,v') = r-1$, and so, $c'=c$.
    In particular, $c$ is both adjacent to some $u' \in S_{r-1}(u,v)$ and to some $v' \in S_r(u,v)$.

    Finally, assume $d(u',c') \ne d(v',c')$.
    Without loss of generality, $d(u',c') < d(v',c')$.
    By Lemma~\ref{lem:alpha1-pties}.\ref{item:alpha1:metric-triangle}, $d(u',c') = 1$ and $d(v',c') = d(u',v') = 2$.
    In this situation, $d(u,u') = d(v,v')+1$.
    Then, $2r-1 = d(u,v) = d(u,u') + d(u',v') + d(v',v) = d(u,u') + 2 + (d(u,u')-1) = 2d(u,u')+1$.
    It implies $d(u,u') = r-1$.
    Hence, $c=c'$.
    Furthermore, vertex $c$ is adjacent to $u' \in S_{r-1}(u,v)$.
\end{proof}

The remainder of this part is devoted to the proof of the following result:

\begin{myproposition}\label{prop:odd-case}
    Let $K$ be a clique in a $\frac 1 2$-hyperbolic graph $G$ such that $\ecc(K) = r+1 \ge 3$, and $\ecc(w) = \ecc(K) = r+1$ for every $w \in K$.    
    There is an $\cO(n+m)$-time procedure that outputs a vertex $a$ such that either $\ecc(a) \le r$, $\ecc(a) = r+1$, or $\ecc(a) \ge 2r$.
    Furthermore, if $\ecc(a) =r+1$ then every vertex of $B_1(K)$ also has eccentricity at least $r+1$.
\end{myproposition}

We remark that since we assume $r \ge 2$, $r < r+1 < 2r$.
Before proving Proposition~\ref{prop:odd-case}, we need to prove a few intermediate results.

\begin{mylemma}\label{lem:slices-clique}
    Let $K$ be a clique of a graph $G$, and let $x$ be such that $d(x,K) \ge 3$.   
    If $G$ is $\frac 1 2$-hyperbolic then the slices $S_2(w,x)$, for $w \in \proj_x(K)$, are comparable for inclusion.
\end{mylemma}
\begin{proof}
    Suppose by contradiction $S_2(w,x)$ and $S_2(w',x)$ are uncomparable for inclusion, for some $w,w' \in \proj_x(K)$.    
    Let $y \in S_2(w,x) \setminus S_2(w',x)$ and let $y' \in S_2(w',x) \setminus S_2(w,x)$.
    By Lemma~\ref{lem:hyp-pties}.\ref{item:leanness}, both $S_2(w,x)$ and $S_2(w',x)$ are cliques.
    Since $K$ is $2$-outergated (Lemma~\ref{lem:cb-graph:cliques}.\ref{item:cb:cliques-2outergated}), $S_2(w,x) \cap S_2(w',x) \ne \emptyset$.
    Therefore, $d(y,y') \le 2$.
    However, by the $\alpha_1$-metric property, $d(y,y') \ge d(y,w) + d(w',y') = 4$.
    A contradiction.
\end{proof}

We combine Lemma~\ref{lem:slices-clique} with Lemma~\ref{lem:alpha1-pties}.\ref{item:triangle-pty}, to prove the following important result:

\begin{mylemma}\label{lem:clique-balls-outergated}
    If $K$ is a clique of a $\frac 1 2$-hyperbolic graph $G$, then $B_2(K)$ is outergated.
\end{mylemma}
\begin{proof}
    Let $x$ be satisfying $d(x,K) \ge 3$.
    By Lemma~\ref{lem:slices-clique}, there exists a $w_x \in \proj_x(K)$ such that $S_2(w,x) \subseteq S_2(w_x,x)$ for every $w \in \proj_x(K)$.
    Then, $\proj_x\left(B_2(K)\right) = S_2(w_x,x)$.
    Let $x^* \in S_3(w_x,x)$ have a maximum number of neighbors in $S_2(w_x,x)$.
    We claim that $x^*$ is an outergate of $x$ with respect to $B_2(K)$.
    Indeed, suppose by contradiction there is a $y \in S_2(w_x,x)$ such that $x^*$ and $y$ are nonadjacent.
    Let $y' \in S_2(w_x,x) \cap N(x^*)$ be arbitrary.
    By Lemma~\ref{lem:hyp-pties}.\ref{item:leanness}, $S_2(w_x,x)$ is a clique.
    Therefore, by Lemma~\ref{lem:alpha1-pties}.\ref{item:triangle-pty}, there exists a $z \in S_3(w_x,x)$ that is adjacent to both $y$ and $y'$.
    Since $x^*,z \in S_1(y',x)$ and $G$ is a CB-graph, $x^*$ and $z$ are adjacent.
    However, the maximality of $x^*$ implies the existence of some $y'' \in S_2(w_x,x)$ such that $y'' \in N(x^*) \setminus N(z)$.
    As a result, $x^*,y'',y,z$ induce a $C_4$, thus contradicting that $G$ is a CB-graph.
\end{proof}

We next summarize ways to restrict the search for a central vertex, using outergates and some convexity arguments:
\begin{mylemma}\label{lem:restrict-search}
    Let $K$ be a clique of a $\frac 1 2$-hyperbolic graph $G$ such that $\ecc(K) = r+1 \ge 3$.
    For every $x \in F(K)$, let $g_K(x)$ be an outergate in $B_2(K)$.
    Let $M = \{ g_K(x) : x \in F(K) \}$.
    The following must hold for any vertex $c \in N(K)$ such that $\ecc(c) \le r$:
    \begin{enumerate}
        \item\label{item:upper-restrict} For every $x \in F(K)$, $c$ has some neighbor $x_c \in \proj_x\left(B_2(K)\right)$ such that $\phi_2(x_c,M)$ is maximized.
        In particular, $d(c,g_K(x)) = 2$, and $\diam(M) \le 4$.
        \item\label{item:lower-restrict} For every $y$ such that $d(y,K) = r$, $N(c) \cap \proj_y(K) \ne \emptyset$.
        Furthermore if \\ $K' = \bigcap\left\{ \proj_y(K) : d(y,K) = r \right\} \ne \emptyset$, then $c \in N(K')$.
    \end{enumerate}
\end{mylemma}
\begin{proof}
    	We give separate proofs (although following similar convexity arguments at some places) for (\ref{item:upper-restrict}) and (\ref{item:lower-restrict}).

        \smallskip
        (\ref{item:upper-restrict}): Let $x \in F(K)$ be arbitrary. 
        Let $K_x = \proj_x(B_2(K))$.
        By Lemma~\ref{lem:slices-clique}, $K_x = S_2(w_x,x)$ for some $w_x \in K$. 
        Therefore, by Lemma~\ref{lem:hyp-pties}.\ref{item:leanness}, $K_x$ is a clique.
        Since $d(x,K) = r+1$, $d(x,c) \le \ecc(c) \le r$, and $c \in N(K)$, we obtain that $c \in S_1(w,x)$ for every $w \in N(c) \cap K$.
        In particular, $c$ has some neighbor in $S_2(w,x) \subseteq K_x$.
        It implies $d(c,g_K(x)) \le 2$, which must be an equality because $d(c,x) - d(g_K(x),x) = 2$.
        Furthermore, since $x$ is arbitrary, $\Phi_2(c,M) = M$.
        Note that it implies that $\diam(M) \le 4$.
        Now, let $b \in K_x$ be maximizing $\phi_2(b,M)$.
        We assume that $b$ and $c$ are nonadjacent (for else, we are done).
        Let $b'$ be any neighbor of $c$ in $K_x$.
        Since $K_x$ is a clique, $d(b,c) = 2$, and $b' \in N(b) \cap N(c)$.
        Recall that $\Phi_2(c,M) = M$.
        Since in addition $G$ is a CB-graph, $\Phi_2(b,M) = \Phi_2(b,M) \cap \Phi_2(c,M) \subseteq \Phi_2(b',M)$.
        Hence, by the maximality of $\phi_2(b,M)$, $\phi_2(b',M) = \phi_2(b,M)$.

        \smallskip
        (\ref{item:lower-restrict}): Let $y$ be an arbitrary vertex such that $d(y,K) = r$.
        We prove that either $\proj_y(K) \subseteq N(c)$, or $K \cap N(c) \subseteq \proj_y(K)$.
        In particular, $\proj_y(K) \cap N(c) \ne \emptyset$.
        For that, let us assume that $\proj_y(K) \not\subseteq N(c)$ (else, we are done).
        Let $w \in \proj_y(K) \setminus N(c)$ be arbitrary.
        Then, $d(w,c) = 2$, and every vertex of $K \cap N(c)$ is on a shortest $wc$-path.
        Since $G$ is a CB-graph, $B_r(y)$ is convex.
        Furthermore, $w,c \in B_r(y)$, and so, $K \cap N(c) \subseteq B_r(y)$.
        Finally, let $K' = \bigcap\left\{ \proj_y(K) : d(y,K) = r \right\} \ne \emptyset$.
        We end up proving that $c \in N(K')$.
        For that, we may further assume that there is no $y$ such that $d(y,K) = r$ and $\proj_y(K) \subseteq N(c)$ (else, we are done).
        In this situation, for every $y$ such that $d(y,K) = r$, $K \cap N(c) \subseteq \proj_y(K)$.
        Since $y$ was chosen arbitrarily, $K \cap N(c) \subseteq \bigcap\left\{ \proj_y(K) : d(y,K) = r \right\} = K'$.
\end{proof}

We need one more technical lemma:
\begin{mylemma}\label{lem:technical-distance-cond}
    Let $x,y,w$ be vertices in a graph $G$ such that $d(x,w) = d(y,w) = 3$, and $d(x,y) \le 3$.
    If $G$ is $\alpha_1$-metric, and $d\left(y,S_1(x,w)\right) \ge 3$, then $S_1(w,x) \cap S_1(w,y) \ne \emptyset$.
\end{mylemma}
        \begin{proof}
            Suppose by contradiction there is a median vertex $v$ for $x,y,w$.
            Since $d(x,w) = d(y,w)$, $d(x,v) = d(y,v)$ holds.
            Since in addition $d(x,y) \le 3$, we obtain $d(x,y) = 2$ and $v \in N(x) \cap N(y)$.
            However, the latter implies that $v \in S_1(x,w)$, thus contradicting our assumption that $d\left(y,S_1(x,w)\right) \ge 3$.
            From now on let $x',y',w'$ be a quasi-median for $x,y,w$.

            Suppose $d(x',w') =  d(y',w') = k$.
            By Lemma~\ref{lem:alpha1-pties}.\ref{item:alpha1:metric-triangle}, $k \le 2$, and $d(x',y') \le k$.
            Since $d(x,w) = d(y,w)$, $d(x,x') = d(y,y')$ holds.
            Since in addition $d(x,y) \le 3$, we get $d(x,x') \le 1$.
            Suppose by contradiction $d(x,x') = 1$.
            Then, $d(x,y) = 2 +d(x',y') \le 3$, and so, $d(x',y') = 1$.
            However, it implies that $x' \in S_1(x,w)$, and $d(y,x') = 2$, thus contradicting our assumption that $d\left(y,S_1(x,w)\right) \ge 3$.
            Therefore, $x=x'$ and $y=y'$.
            Since $d\left(y,S_1(x,w)\right) \ge 3$, we must have $d(x',y') = 2$.
            In particular, $d(x',w') =  d(y',w') = d(x',y') = 2$.
            Note that in this situation, $I^o(x',w') \subseteq S_1(x,w)$.
            However, as $G$ is a CB-graph, Theorem $7.1$ of~\cite{CBgraphs} implies the following \emph{strong equilateral} condition: $d(x',I(y',z')) = d(y',I(x',z')) = d(z',I(x',y')) = 2$.
            In particular, $d(y,S_1(x,w)) \le d(y,I^o(x',w')) = 2$. A contradiction.

            Therefore, $d(x',w') \ne d(y',w')$.
            (See Fig.~\ref{fig:odd-case} for $r=2$ and $w=z$).
            By Lemma~\ref{lem:alpha1-pties}.\ref{item:alpha1:metric-triangle}, $d(x',y') = 2$.
            Since $d(x,y) \le 3$, either $x = x'$ or $y=y'$.
            Since in addition $d(x,w) = d(y,w) = 3$, and also by Lemma~\ref{lem:alpha1-pties}.\ref{item:alpha1:metric-triangle}, $d(x',w') \le 2$ and $d(y',w') \le 2$, $w' \ne w$ holds.
            Hence, we obtain $S_1(w,w') \subseteq S_1(w,x) \cap S_1(w,y)$.
        \end{proof}

The proof of Proposition~\ref{prop:odd-case} now follows from the analysis of Algorithm~\ref{alg:odd-case}.
See also Fig.~\ref{fig:odd-case-algorithm} for an illustration.
The implementation of the algorithm can be reduced to $\cO(1)$ calls to BFS (either from a start vertex, or from a start subset), and to the procedure of Lemma~\ref{lem:intermediate-computations}.
Therefore, the running time is linear.

\setlength\columnsep{20pt}
\begin{algorithm}[!h]
	\DontPrintSemicolon
    \SetAlgoLined
    \begin{multicols}{2}
	\SetKwInOut{Input}{Input}
	\SetKwInOut{Output}{Output}
	\SetKw{Continue}{continue}
	\SetKw{Break}{break}
    \scriptsize
    
	\Input{a graph $G=(V,E)$, and a clique $K$ of $G$.}
	\Output{a vertex $a$.}

    \BlankLine
    \BlankLine

    $r := \ecc(K)-1$\;

    \BlankLine
    $F'(K) := \{ z \in V : d(z,K) = r \}$\;
    \If{$F'(K) \ne \emptyset$}
    {
         compute the values $\psi_K(w,F'(K)), \ w \in K$\;
         $w_0 := $ a vertex of $K$ s.t. $\psi_K(w_0,F'(K))$ is maximized\;
         \eIf{$\psi_K(w_0,F'(K)) < |F'(K)|$}
         {
            \Return an arbitrary vertex of $F(w_0) \cap F'(K)$ {\color{blue}\it //$\ecc(a) \ge 2r$}\;
         }
         {
            $K := \{ w \in K : \Psi_K(w,F'(K)) = F'(K) \}$\;
         }
    }

    \BlankLine
    compute outergates $g_K(x), \ x \in F(K)$ w.r.t. $B_2(K)$\;
    $M := \{ g_K(x) : x \in F(K) \}$\;
    $x_0 := $ an arbitrary vertex of $F(K)$\;

    \BlankLine
    \If{$\exists x \in F(K) \ \text{s.t.} \ d(g_K(x_0),g_K(x)) \ge 5$}
    {
        \Return an arbitrary vertex of $K$ {\color{blue}\it //$\ecc(a) = r+1$}\;
    }

    \BlankLine
    \If{$\exists x \in F(K) \ \text{s.t.} \ d(g_K(x_0),g_K(x)) = 4$}
    {
        $U := S_2(g_K(x_0),g_K(x))$\;
        $b :=$ a vertex of $U$ s.t. $\ecc(b)$ is minimized\;
        \eIf{$\ecc(b) \le r$}
        {
            \Return $b$ {\color{blue}\it //$\ecc(a) \le r$}\;
        }
        {
            \Return an arbitrary vertex of $K$ {\color{blue}\it //$\ecc(a)=r+1$}\;
        }
    }

    \BlankLine
    {\color{blue}\it /*$\forall x  \in F(K),  d(g_K(x_0),g_K(x)) \le 3$*/}\;
    $L := \{ s \in \proj_{x_0}(B_2(K)) : \phi_2(s,M) \ \text{is maximized} \}$\;
    $s_0 :=$ an arbitrary vertex of $L$\;
    $M_0' := \{ y \in M : d(y,\proj_{x_0}(B_2(K))) \le 2 \}$\;
    \If{$M_0' \not\subseteq B_2(s_0)$}
    {
        $x_0 :=$ any vertex of $F(K)$ s.t. $g_K(x_0) \in M_0' \setminus B_2(s_0)$\;
        {\bf goto} line $15$ \;
    }
    \BlankLine
    {\color{blue}\it /*$\forall s \in L, \Phi_2(s,M) = M_0'$*/}\;
    \If{$\ecc(L) < r$}
    {
        \Return an arbitrary vertex of $L$ {\color{blue}\it //$\ecc(a) \le r$}\;
    }
    \If{$\ecc(L) = r$}
    {
        $s_0 :=$ a vertex of $L$ s.t. $\psi_L(s_0)$ is maximized\;
        \eIf{$\ecc(s_0) = r$}
        {
            \Return $s_0$ {\color{blue}\it //$\ecc(a) \le r$}\;
        }
        {
            \Return an arbitrary vertex of $F(s_0)$ {\color{blue}\it //$\ecc(a) \ge 2r$}\;
        }
    }
    \BlankLine
    {\color{blue}\it /*$\ecc(L) \ge r+1$*/}\;
    $z_L :=$ an arbitrary vertex of $F(L)$\;
    \If{$d(z_L,N(L)\cap N(K)) \ge r+1$}
    {
        \Return an arbitrary vertex of $K$ {\color{blue}\it //$\ecc(a) = r+1$}\;
    }
    $c_L :=$ an arbitrary vertex of $N(L) \cap N(K) \cap B_r(z_L)$\; 
    \If{$M_0' \not\subseteq B_2(c_L)$}
    {
        \Return $z_L$ {\color{blue}\it //$\ecc(a) \ge 2r$}\;
    }
    \BlankLine
    {\color{blue}\it /*$\forall y \in M_0', d(c_L,y) = 2$*/}\;
    $w_L :=$ an arbitrary vertex of $K \cap N(c_L)$\;
    $W_L := S_1(w_L,x_0)$\;
    $v_L :=$ a vertex of $W_L$ s.t. $\phi_2(v_L,M)$ is maximized\;
    \If{$\phi_2(v_L,M) < |M|$}
    {
        $x_0 :=$ any vertex of $F(K)$ s.t. $g_K(x_0) \in M \setminus B_2(v_L)$\;
        {\bf goto} line 15\;
    }
    \BlankLine
     {\color{blue}\it /*$\forall x \in F(K), d(v_L,g_K(x)) = 2$*/}\;
     \eIf{$\ecc(v_L) = r$}
     {
        \Return $v_L$ {\color{blue}\it //$\ecc(a) \le r$}\;
     }
     {
        \Return an arbitrary vertex of $F(v_L)$ {\color{blue}\it //$\ecc(a) \ge 2r$}\;
     }
    \BlankLine
    \BlankLine
    \end{multicols}
    \caption{The \texttt{process-clique} procedure, see Prop.~\ref{prop:odd-case}.}
    \label{alg:odd-case}
\end{algorithm}

        \begin{proof}[Proof of Proposition~\ref{prop:odd-case}]
            The result follows from the analysis of Algorithm~\ref{alg:odd-case}.
            
    \medskip
    \noindent
    {\bf Correctness.}
    First, we force the clique $K$ to satisfy the following Property $(\alpha)$:
    $$\forall z \in V \setminus F(K), \forall w \in K, d(z,w) \le r$$
    The intuition goes as follows: assume that we can find a vertex $v_L \in N(K)$ such that $F(K) \subseteq B_r(v_L)$.
    Then, either $\ecc(v_L) \le r$, or $\ecc(v_L) = r+1$ and $F(v_L) \subseteq V \setminus F(K)$.
    Furthermore, in the latter situation, assuming Property $(\alpha)$ we can consider any edge between $v_L$ and $K$; by applying the $\alpha_1$-metric property to this edge, we can prove that every vertex of $F(v_L)$ has eccentricity $\ge 2r$ (lines 65--69). 
    In order to enforce Property $(\alpha)$, we use the set of vertices $F'(K) = \{ z \in V : d(z,K) = r \}$ (lines 2--11).
    If $F'(K) = \emptyset$, then Property $(\alpha)$ already holds.
    Otherwise, either some vertex of $K$ is at distance at most $r$ to every vertex of $F'(K)$, in which case we can restrict $K$ to $K' = \{ w \in K : \psi_K(w,F'(K)) = F'(K) \}$ (see Lemma~\ref{lem:restrict-search}.\ref{item:lower-restrict}); or we can output a vertex of eccentricity $\ge 2r$ (Lemma~\ref{lem:larger-x}).
    More specifically, let $w_0 \in K$ be maximizing $\psi_K(w_0,F'(K))$.
    If $\psi_K(w_0,F'(K)) < |F'(K)|$, then the algorithm outputs any vertex $a \in F(w_0) \cap F'(K)$. 
    Since $\ecc(w_0) = r+1$, we have that $F(w_0) \cap F'(K) = F'(K) \setminus \Psi_K(w_0,F'(K))$. Therefore, by Lemma~\ref{lem:larger-x}, $\ecc(a) \ge 2r$.
    Else, we replace $K$ with $K' = \{ w \in K : \psi_K(w,F'(K)) = F'(K) \}$.
    In doing so, $\ecc(K') = r+1$, and $F(K') = F(K)$.
    Furthermore, by Lemma~\ref{lem:restrict-search}.\ref{item:lower-restrict}, any vertex $c \in N(K)$ such that $\ecc(c) \le r$ must be contained in $N(K')$.
    Therefore, we assume for the remainder of the proof Property $(\alpha)$ holds.

    For every $x \in F(K)$, let $g_K(x)$ be an outergate of $x$ in $B_2(K)$, whose existence follows from Lemma~\ref{lem:clique-balls-outergated}.
    Let $M = \{ g_K(x) : x \in F(K) \}$.
    Roughly, most of the algorithm (lines 12--59) is devoted to the search for a vertex $v_L$ such that $M \subseteq B_2(v_L)$.
    Such a vertex, if any, may not be central, but it satisfies $F(K) \subseteq B_r(v_L)$.
    We start presenting a rough sketch of the sub-procedure for searching $v_L$ (each component of the sub-procedure will be detailed in the remainder of the proof).
    For that, we fix an arbitrary $x_0 \in F(K)$.
    Let $M_0'$ be the subset of vertices in $M$ that are at distance at most $2$ to $\proj_{x_0}(B_2(K))$.
    We start searching in $\proj_{x_0}(B_2(K))$ for a vertex $s_0$ such that $M_0' \subseteq B_2(s_0)$ (lines 28--34).
    Let $L$ be the subset of all such vertices $s_0$.
    Then, using an arbitrary $z_L$ such that $d(z,L) > r$, we go closer to $K$ by searching for a $c_L \in N(K) \cap N(L)$ such that $c_L \in B_r(z_L)$, and $M_0' \subseteq B_2(c_L)$ (see Lemma~\ref{lem:restrict-search}.\ref{item:upper-restrict}). 
    Doing so, let $w_L \in N(c_L) \cap K$ be arbitrary, and let $W_L = S_1(w_L,x_0)$. 
    We can prove that $d(g_K(x),W_L) \le 2$ for every $x \in F(K)$ (see Claim~\ref{claim:dist-2-clique} below).
    Therefore, we search for $v_L$ in $W_L$, which by Lemma~\ref{lem:hyp-pties}.\ref{item:leanness} is a clique.

    Note that in what follows, we can assert at various places that there is no $c \in N(K)$ such that $\ecc(c) \le r$.
    In this situation, the algorithm outputs any vertex $a \in K$, and $\ecc(a) = r+1$ (see lines 16, 24 \& 50).

    The sub-procedure for searching $v_L$ may fail if $\diam(M) \ge 4$.
    More specifically, sometimes we cannot compute an intermediate vertex with the desired properties, in which case we rather find some $x,x' \in F(K)$ such that $d(g_K(x),g_K(x')) \ge 4$.
    In these situations, we replace $x_0$ with one of $x$ or $x'$, then we apply the sub-procedure of lines 15--26.
    The analysis of this case is presented in the following claim:
    \begin{myclaim}\label{claim:x0}
        If we find an $x_0 \in F(K)$ such that for some other $x \in F(K)$, $d(g_K(x_0),g_K(x)) \ge 4$, then the algorithm either outputs a vertex $a$ such that $\ecc(a) \le r$, or it correctly asserts that there is no vertex $c \in N(K)$ such that $\ecc(c) \le r$.        
    \end{myclaim}
    If $d(g_K(x_0),g_K(x)) \ge 5$, then $\diam(M) \ge 5$.
    By Lemma~\ref{lem:restrict-search}.\ref{item:upper-restrict}, there is no $c \in N(K)$ such that $\ecc(c) \le r$.
    If $d(g_K(x_0),g_K(x)) = 4$, then let $U = S_2(g_K(x_0),g_K(x))$. 
    Either there is a vertex $a \in U$ such that $\ecc(a) \le r$, or by Lemma~\ref{lem:restrict-search}.\ref{item:upper-restrict} no such a vertex exists in $N(K)$. $\diamond$ 

    \smallskip
    Initially, we fix an arbitrary $x_0 \in F(K)$.
    By Claim~\ref{claim:x0}, we may assume $d(g_K(x_0),g_K(x)) \le 3$ for every $x \in F(K)$.
    Then, let $L = \{ s \in \proj_{x_0}(B_2(K)) : \phi_2(s,M) \ \text{is maximized} \}$.
 
    We consider the subset $M_0' = \{ y \in M : d(y,\proj_{x_0}(B_2(K))) \le 2\}$.
    For every $s \in L, \Phi_2(s,M) \subseteq M_0'$ holds.
    In particular, the following result is true for any $s_0 \in L$:
    \begin{myclaim}\label{claim:s0}
        Assume $\phi_2(s_0,M) < |M_0'|$.
        For every $x_0' \in F(K)$ with $g_K(x_0') \in M_0' \setminus B_2(s_0)$, there exists an $x_1 \in F(K)$ such that $d(g_K(x_0'),g_K(x_1)) \ge 4$.
    \end{myclaim}
    Let $s_1 \in \proj_{x_0}(B_2(K))$ be such that $d(g_K(x_0'),s_1) \le 2$.
    As $s_0 \in L$, $\phi_2(s_0,M)$ is maximum within the vertices of $\proj_{x_0}(B_2(K))$.
    In particular, $\phi_2(s_1,M) \le \phi_2(s_0,M)$.
    Then, as $g_K(x_0') \in B_2(s_1) \setminus B_2(s_0)$, there exists an $x_1 \in F(K)$ such that $g_K(x_1) \in B_2(s_0) \setminus B_2(s_1)$.
    Furthermore, we recall that by the combination of Lemma~\ref{lem:alpha1-pties}.\ref{item:leanness} and Lemma~\ref{lem:slices-clique}, $\proj_{x_0}(B_2(K))$ is a clique.
    The latter implies $d(s_0,g_K(x_1)) = d(s_1,g_K(x_0')) = 2$.
    By the $\alpha_1$-metric property, $d(g_K(x_0'),g_K(x_1)) \ge d(g_K(x_0'),s_1)+d(s_0,g_K(x_1)) = 4$. $\diamond$
    
    In this situation, we replace $x_0$ with $x_0'$, and then we are done by Claim~\ref{claim:x0}.

    \smallskip
    Therefore, from now on we assume $\Phi_2(s,M) = M_0'$ for every $s \in L$.
    There are several cases to be distinguished from each other, which depends on the value of $\ecc(L)$.
    Note that in what follows, in a few technical situations, the search for a vertex $v_L$ may fail even if $\diam(M) \le 3$.
    However, when it happens, we can either directly output a vertex $a$ such that $\ecc(a) \le r$ ($a$ may not be in $N(K)$, see lines 37 \& 42), or $\ecc(a) \ge 2r$ (using Lemma~\ref{lem:larger-x}, see lines 44 \& 54); or we can assert that there is no $c \in N(K)$ such that $\ecc(c) \le r$.

    If $\ecc(L) < r$, then the algorithm outputs any vertex $a \in L$.
    Doing so, $\ecc(a) \le r$.

    Assume now that $\ecc(L) = r$.
    We reuse the same strategy as in~\cite{ChDr}.
    More specifically, let $s_0 \in L$ be maximizing $\psi_L(s_0)$.
    If $\ecc(s_0) = r$, then the algorithm simply outputs $a=s_0$.
    Otherwise, the algorithm outputs any vertex $a \in F(s_0)$, and by Corollary~\ref{cor:larger-x}, $\ecc(a) \ge 2r$.

    From now on, $\ecc(L) \ge r+1$.
    Let $z_L \in F(L)$ be arbitrary.
    If there is no vertex of $N(L) \cap N(K)$ at distance no more than $r$ to $z_L$, then by Lemma~\ref{lem:restrict-search}.\ref{item:upper-restrict} we can assert that there is no $c \in N(K)$ such that $\ecc(c) \le r$.
    Hence, we shall assume in what follows the existence of some $c_L \in N(L) \cap N(K) \cap B_r(z_L)$.
    Furthermore,
    \begin{myclaim}\label{claim:zLcL}
        Either $M_0' \subseteq B_2(c_L)$, or $\ecc(z_L) \ge 2r$.        
    \end{myclaim}
    To prove the claim, assume the existence of an $x \in F(K)$ such that $g_K(x) \in M_0' \setminus B_2(c_L)$.
    As $g_K(x)$ is an outergate of $x$ with respect to $B_2(K)$, for each vertex $w \in K$, $d(c,g_K(x)) = 2$ holds for all vertices $c \in S_1(w,x)$.
    The latter implies that $c_L$ is not on a shortest $wx$-path, for any $w \in K$.
    In this situation, as $d(x,K) = r+1$ and $c_L \in N(K)$, we obtain $d(x,c_L) \ge r+1$.
    Then, let $s \in L \cap N(c_L)$ be arbitrary.
    Note that $d(s,x) \le d(s,g_K(x))+d(g_K(x),x) \le 2 + r-2 = r$.
    As $d(x,c_L) \ge r+1$, it implies $d(s,x) = d(c_L,x)-1 = r$.
    Therefore, by the $\alpha_1$-metric property, $\ecc(z_L) \ge d(z_L,x) \ge d(z_L,c_L)+d(s,x) = 2r$. $\diamond$
    
    If $\ecc(z_L) \ge 2r$, then the algorithm can output $a=z_L$.
    Hence, by Claim~\ref{claim:zLcL}, we can assume for the remainder of the proof $M_0' \subseteq B_2(c_L)$.

    \smallskip
    The following claim is the gist of our approach:
    \begin{myclaim}\label{claim:dist-2-clique}
        Let $w_L \in N(c_L) \cap K$.
        For every $x \in F(K)$, $d\left(g_K(x),S_1(w_L,x_0)\right) = 2$.
    \end{myclaim}
    Indeed, let $x \in F(K)$ be arbitrary.
    Since $S_1(w_L,x_0) \subseteq N(K)$, and $d(g_K(x),K) = 3$, $d\left(g_K(x),S_1(w_L,x_0)\right) \ge 2$. 
    In order to prove this is an equality, it suffices to prove $S_1(w_L,x) \cap S_1(w_L,x_0) \ne \emptyset$. 
    If $g_K(x) \in M_0'$ then this is true because $c_L \in S_1(w_L,x_0) \cap S_1(w_L,x)$. 
    Otherwise, $d(g_K(x),w_L) = d(g_K(x_0),w_L) = 3$, $d(g_K(x),g_K(x_0)) \le 3$, and $d\left(g_K(x),S_1(g_K(x_0),w_L)\right) \ge d(g_K(x),\proj_{x_0}(B_2(K))) \ge 3$.
    Then, the result follows from Lemma~\ref{lem:technical-distance-cond}. $\diamond$

    \smallskip
    Let $W_L = S_1(w_L,x_0)$.
    By Claim~\ref{claim:dist-2-clique}, for every $v \in W_L$, $\Phi_2(v,M) = \Psi_{W_L}(v,M)$.
    Hence, let $v_L \in W_L$ be maximizing $\phi_2(v_L,M)$.
    If $\Phi_2(v_L,M) \ne M$ then let $x_0' \in F(K)$ be satisfying $d(g_K(x_0'),v_L) \ge 3$.
    By Lemma~\ref{lem:larger-x}, there is an $x \in F(K)$ such that $d(g_K(x),g_K(x_0')) \ge 4$. Therefore, by Claim~\ref{claim:x0}, we are done replacing $x_0$ with $x_0'$.
    Otherwise, $M \subseteq B_2(v_L)$, and so, $F(K) \subseteq B_r(v_L)$.
    Assume further $\ecc(v_L) > r$ (else, the algorithm outputs $a=v_L$).
    Let $a \in F(v_L)$ be arbitrary, and let $w \in N(v_L) \cap K$.
    Since $F(K) \subseteq B_r(v_L)$, we get $a \in V \setminus F(K)$.
    Furthermore, by Property $(\alpha)$, $d(w,a) \le r$.
    By the $\alpha_1$-metric property, $\ecc(a) \ge d(a,x_0) \ge d(a,w)+d(v_L,x_0) = 2r$.
    As a result, we are done outputting $a$.

    \medskip
    \noindent
    {\bf Runtime analysis.}
    We can compute $r,F'(K) \ \text{and} \ F(K)$ by running a BFS with start subset $K$.
    Assume further $F'(K) \ne \emptyset$ (lines 3--11).
    Then, we can apply the procedure of Lemma~\ref{lem:intermediate-computations}.\ref{item:compute-shadows} to compute all values $\psi_K(w,F'(K)), \ w \in K$.
    During this process, we can also compute the vertex $w_0$ and the sub-clique $K'$.
    If $\psi_K(w_0,F'(K)) < |F'(K)|$, then we compute $F'(K) \cap F(w_0)$ by running a BFS from $w_0$. 
    Otherwise, we can apply the procedure of Lemma~\ref{lem:compute-outergates} to compute the outergates $g_K(x), \ x \in F(K)$ with respect to $B_2(K)$; let $M = \{g_K(x) : x \in F(K)\}$ be the resulting subset.
    
    For an arbitrary $x_0 \in F(K)$, we can compute some $x \in F(K)$ such that $d(g_K(x_0),g_K(x))$ is maximized by running a BFS from $g_K(x_0)$.
    If $d(g_K(x_0),g_K(x)) = 4$, then we compute $U=S_2(g_K(x_0),g_K(x))$ by running another BFS from $g_K(x)$; by Lemma~\ref{lem:hyp-pties}.\ref{item:leanness}, $U$ is a clique, and therefore, we can compute a vertex $b \in U$ such that $\ecc(b)$ is minimized by applying the procedure of Lemma~\ref{lem:intermediate-computations}.\ref{item:compute-ecc-clique}.
    Let us assume $d(g_K(x_0),g_K(x)) \le 3$ for the remainder of the analysis (line 27).
    Note that at an ulterior step of the algorithm, we may replace $x_0$ with some other $x_0' \in F(K)$ (lines 32 \& 61).
    However, our prior correctness analysis implies that in such cases, there exists some $x' \in F(K)$ such that $d(g_K(x_0'),g_K(x')) \ge 4$.
    In this situation, we repeat once lines 15--26 of the algorithm, and then we halt.
   
    Since we performed a BFS with start subset $K$, the ball $B_2(K)$ is known to us.
    In particular, we can compute $\proj_{x_0}(B_2(K))$ by scanning the neighbors of $g_K(x_0)$, and taking their intersection with $B_2(K)$.
    Since by the combination of Lemma~\ref{lem:hyp-pties}.\ref{item:leanness} and Lemma~\ref{lem:slices-clique}, $\proj_{x_0}(B_2(K))$ is a clique, the subset $L$ can be computed by applying the procedure of Lemma~\ref{lem:intermediate-computations}.\ref{item:compute-potentials}.
    Furthermore, we can compute the subset $M_0'$ by running a BFS with start subset $\proj_{x_0}(B_2(K))$.
    For an arbitrary $s_0 \in L$, we can compute $\Phi_2(s_0,M)$, and compare it to $M_0'$, by running a BFS from $s_0$.
    Then, let us analyze lines 35--46 of the algorithm.
    The value $\ecc(L)$ can be computed by running a BFS with start subset $L$.
    If $\ecc(L) = r$ then since $L$ is a clique, a vertex  $s_0$ such that $\psi_L(s_0)$ is maximized can be computed by applying the procedure of Lemma~\ref{lem:intermediate-computations}.\ref{item:compute-shadows}.
    We can compute $\ecc(s_0)$, and a vertex of $F(s_0)$, by running another BFS from $s_0$.
    We now analyze lines 47--55 of the algorithm.
    The vertex $z_L \in F(L)$ was obtained as a byproduct of the BFS with start subset $L$.
    Then, by running a BFS with start subset $N(L) \cap N(K)$, we can either compute a vertex $c_L \in N(L) \cap N(K) \cap B_r(z_L)$, or conclude that no such vertex exists.
    Running another BFS from $c_L$, we can compute $\Phi_2(c_L,M)$, and compare it with $M_0'$.

    We compute $w_L \in N(c_L) \cap K$ by scanning $K$ and $N(c_L)$.
    Afterwards, we compute $W_L = S_1(w_L,x_0)$ by scanning the neighbors of $w_L$ and the neighbors of every vertex of $\proj_{x_0}(B_2(K))$.
    Since $G$ is a CB-graph, $W_L$ is a clique.
    Hence, to compute a vertex $v_L \in W_L$ such that $\phi_2(v_L,M)$ is maximized, we can apply the procedure of Lemma~\ref{lem:intermediate-computations}.\ref{item:compute-potentials}.
    Finally, we run a BFS from $v_L$ in order to compute $\Phi_2(v_L,M)$, $\ecc(v_L)$, and a vertex of $F(v_L)$.

    Overall, the total runtime is in $\cO(n+m)$.
        \end{proof}

      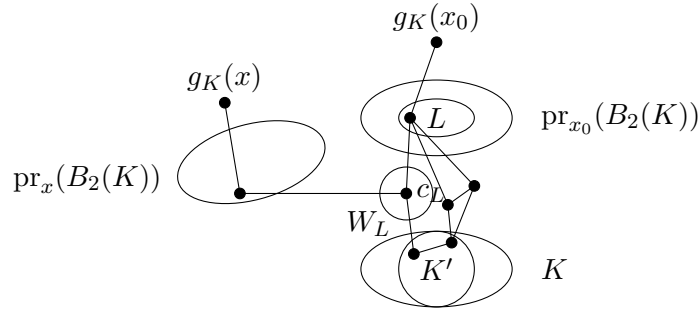
\begin{figure}[!h]
    \centering
    \begin{tikzpicture}
        \draw (0,0) ellipse(1cm and .5cm);
        \draw (0,0) circle(.5cm) node{$K'$};
        \draw (1.25,0) node[anchor=west]{$K$};
        \draw (0,2) ellipse(1cm and .5cm);
        \draw (1.25,2) node[anchor=west]{$\proj_{x_0}(B_2(K))$};
        \draw (0,2) ellipse(.5cm and .25cm) node{$L$};
        \draw[rotate=15] (-2,2) ellipse(1cm and .5cm) node{};
        \draw (-.4,1) circle(.35cm);
        \draw (-.9,.6) node{$W_L$};
        \draw (-3.5,1.2) node[anchor=east]{$\proj_x(B_2(K))$};
        \filldraw (-.3,.2) circle (2pt) node{};
        \filldraw (.2,.35) circle (2pt) node{};
        \filldraw (-.35,2) circle (2pt) node{};
        \filldraw (-.4,1) circle (2pt) node[anchor=west]{$c_L$};
        \filldraw (.15,.85) circle (2pt) node{};
        \filldraw (.5,1.1) circle (2pt) node{};
        \filldraw (-2.6,1) circle (2pt) node{};
        \filldraw (0,3) circle (2pt) node[anchor=south]{$g_K(x_0)$};
        \filldraw (-2.8,2.2) circle (2pt) node[anchor=south]{$g_K(x)$};
        \draw (-.3,.2) -- (.2,.35);
        \draw (.15,.85) -- (.5,1.1);
        \draw (-.3,.2) -- (-.4,1);
        \draw (.15,.85) -- (.2,.35) -- (.5,1.1);
        \draw (.15,.85) -- (-.35,2) -- (.5,1.1);
        \draw (-.35,2) -- (-.4,1);
        \draw (-.4,1) -- (-2.6,1);
        \draw (0,3) -- (-.35,2);
        \draw (-2.6,1) -- (-2.8,2.2);
    \end{tikzpicture}
    \caption{An illustration of the proof of Prop.~\ref{prop:odd-case}.}
    \label{fig:odd-case-algorithm}
    \end{figure}

\subsection{The algorithm}\label{ssec:alg-analysis}

\begin{proof}[Proof of Theorem~\ref{thm:center:half-hyp}]
     The result follows from the analysis of Algorithm~\ref{alg:center}.

            \begin{algorithm}[!h]
	\DontPrintSemicolon
    \SetAlgoLined
     \begin{multicols}{2}
	\SetKwInOut{Input}{Input}
	\SetKwInOut{Output}{Output}
	\SetKw{Continue}{continue}
	\SetKw{Break}{break}
    \scriptsize

	\Input{a $\frac 1 2$-hyperbolic graph $G=(V,E)$.}
	\Output{a central vertex of $G$.}

	\BlankLine
    $u,v := $ a pair of mutually distant vertices\;
    $r := \left\lceil \frac {d(u,v)} 2 \right\rceil$\;

    \BlankLine
    \BlankLine
    {\it\color{blue} /**even case**/}\;
    \BlankLine
    \If{$d(u,v) = 2r$}
    {
        $K := S_r(u,v)$\;
        $c :=$ a vertex of $K$ s.t. $\ecc(c)$ is minimized\;
        \Return $c$\;
    }

    \BlankLine
    \BlankLine
    {\it\color{blue}/**odd case: $d(u,v) = 2r-1$**/}\;

        \BlankLine
        \If{$r=1$}
        {
            \Return $u$\;
        }

        \BlankLine
        {\it\color{blue}/*$r \ge 2$*/}\;
        $K := S_{r-1}(u,v)$\;
        first\_slice\_scanned $:=$ TRUE\; 

        \If{$e(K) = r$}
        {
            $c :=$ a vertex of $K$ s.t. $\psi_K(c)$ is maximized\;
            \eIf{$e(c) = r$}
            {
                \Return $c$\;
            }
            {
                $u :=$ an arbitrary vertex of $F(c)$\;
                $v :=$ an arbitrary vertex of $F(u)$\;
                {\bf goto} line 3\;
            }
        }
        \BlankLine
        {\it\color{blue}/*$\ecc(K) = r+1$*/}\;
        $a :=$ {\tt process-clique}$(G,K)$\;
        \BlankLine
        \If{$\ecc(a) = r$}
        {
            \Return $a$\;
        }
        \BlankLine
        \If{$\ecc(a) = 2r$}
        {
            $u := a$\;
            $v :=$ an arbitrary vertex of $F(a)$\;
            {\bf goto} line 3\;
        }

        \BlankLine
        {\it\color{blue}/*$\ecc(a) = r+1$*/}\;
        \eIf{first\_slice\_scanned}
        {
            $K := S_r(u,v)$\;
            first\_slice\_scanned $:=$ FALSE\;
            {\bf goto} line 16\;
        }
        {
            \Return $a$\;
        }
	 \BlankLine
	 \BlankLine
  \end{multicols}
	 \caption{Computing a central vertex in a $\frac 1 2$-hyperbolic graph.}
  \label{alg:center}
\end{algorithm}

     \medskip
     \noindent
     {\bf Correctness.}
     Let $u,v$ be mutually distant vertices.
     If $d(u,v) = 2r$ (even case), by Corollary~\ref{cor:even-case}, every vertex of minimum eccentricity within $S_r(u,v)$ is central.
     Hence, we shall assume for the remainder of the proof $d(u,v) = 2r-1$ (odd case).
     By Lemma~\ref{lem:ecc-ub}, we are left deciding whether $rad(G) = r$ or $rad(G) = r+1$, and returning a central vertex.
     If $r=1$, then both $u$ and $v$ are universal. In particular, both $u$ and $v$ are central, and so we can return any of those vertices.
     From now on, we shall assume $r \ge 2$.
     By Lemma~\ref{lem:odd-case:locate-center}, a vertex of eccentricity at most $r$, if any, must be contained either in $B_1\left(S_{r-1}(u,v)\right)$, or in $B_1\left(S_r(u,v)\right)$.
     We first consider $K = S_{r-1}(u,v)$, which by Lemma~\ref{lem:hyp-pties}.\ref{item:leanness} is a clique.
     There are two cases:

     If $\ecc(K) = r$, then we follow the same strategy as in~\cite{ChDr}. More specifically, let $c \in K$ be maximizing $\psi_K(c)$. If $\ecc(c) = r$, then $c$ is central. Otherwise, let $u' \in F(c)$. By Corollary~\ref{cor:larger-x}, $\ecc(u') \ge 2r$.
     Furthermore, by Lemma~\ref{lem:hyp-pties}.\ref{item:hyp:furthest-vertex}, $\ecc(u) = \ecc(v)  = 2r-1 \ge \diam(G)-1$.
     Therefore, $\ecc(u') = \diam(G) = 2r$.
     Let $v' \in F(u')$ be arbitrary.
     We replace $u,v$ with the diametral pair $u',v'$. By doing so, we are back to the even case.

    Otherwise, by Lemma~\ref{lem:ecc-ub}, $\ecc(K) = r+1$; furthermore, also by Lemma~\ref{lem:ecc-ub}, $\ecc(w) = r+1$ for every $w \in K$.
    Let $a$ be the vertex returned by the {\tt process-clique} procedure of Proposition~\ref{prop:odd-case} (Algorithm~\ref{alg:odd-case}).
    If $\ecc(a) = r$, then $a$ is central.
    Else, if $\ecc(a) \ge 2r$, then we conclude as before $\ecc(a) = 2r = \diam(G)$.
    We replace $u$ and $v$ with $a$ and an arbitrary vertex of $F(a)$. By doing so, we are back to the even case.
    Otherwise, $\ecc(a) = r+1$, and by Proposition~\ref{prop:odd-case}, there is no vertex of eccentricity $r$ within $B_1(K)$.
    In the latter sub-case, we next consider $K = S_r(u,v)$.
    We proceed exactly as before, with the exception of the last sub-case: when $\ecc(K) = r+1$, and the vertex $a$ returned by the {\tt process-clique} procedure has eccentricity $r+1$. 
    Then, by Proposition~\ref{prop:odd-case} (applied twice), there is no vertex of eccentricity $r$ within $B_1\left(S_{r-1}(u,v)\right) \cup B_1\left(S_r(u,v)\right)$.
    By Lemma~\ref{lem:odd-case:locate-center}, $\rad(G) = r+1$, and $a$ is central.

     \medskip
     \noindent
     {\bf Runtime analysis.}
     We apply Lemma~\ref{lem:hyp-pties}.\ref{item:mutually-distant_vertices} to compute a pair $u,v$ of mutually distant vertices.
     Then, we run two BFSs: from $u$ and $v$, respectively.
     In the even case, since by Lemma~\ref{lem:hyp-pties}.\ref{item:leanness} $S_r(u,v)$ is a clique, a vertex of minimum eccentricity within the slice can be computed by applying the procedure of Lemma~\ref{lem:intermediate-computations}.\ref{item:compute-ecc-clique}.
     We now focus on the odd case, and we further assume $r \ge 2$ (the sub-case $r=1$ is trivial).
     Since we apply the same procedure to $S_{r-1}(u,v)$ and $S_r(u,v)$, it suffices to detail the implementation of this procedure for a fixed clique $K$.
     We run a BFS with start subset $K$ to compute $\ecc(K)$ and $F(K)$.
     If $\ecc(K) = r$, then we compute the vertex $c$ by applying Lemma~\ref{lem:intermediate-computations}.\ref{item:compute-shadows};
     then, we compute $\ecc(c)$ and $F(c)$ by running a BFS from $c$.
     If $\ecc(K) = r+1$, then we apply Proposition~\ref{prop:odd-case}; 
     the eccentricity of its output $a$ can be computed from a final BFS.
     The overall running-time is therefore in $\cO(n+m)$.
\end{proof}

\section{Hardness results}\label{sec:hard}

We complete Sec.~\ref{sec:alg} with two hardness results: one for the recognition of $\frac 1 2$-hyperbolic graphs, and another for the {\sc Center} problem on $1$-hyperbolic graphs.

\smallskip
\noindent
{\bf Recognition of $\frac 1 2$-hyperbolic graphs.} (Theorem~\ref{thm:hardness:half-hyp}, see Sec.~\ref{sec:hardness:half-hyp} below).
The Strong Exponential-Time Hypothesis (SETH) posits that for any $\varepsilon > 0$, there is a $k$ such that $k$-{\sc SAT} on $n$ variables cannot be solved in $\cO\left((2-\varepsilon)^n\right)$ time~\cite{SETH}.
The Orthogonal-Vector problem ({\sc OV}) takes as input two families $A$ and $B$ of $n$ sets over some universe $U$, and it asks whether there exist $a \in A$, $b \in B$ s.t. $a \cap b = \emptyset$.
By~\cite{OV}, under SETH we cannot solve {\sc OV} in $\cO(n^{2-\varepsilon})$ time, for any $\varepsilon > 0$.

In~\cite{hypChordal}, a characterization of the $\frac 1 2$-hyperbolic chordal graphs is proved, via two forbidden isometric subgraphs. For split graphs, the characterization can be reduced to one such subgraph $H_2$. 
To see the connection with {\sc OV}, we stress that $H_2$ contains two vertices $u$ and $v$ at distance three, or equivalently such that $B_1(u)$ and $B_1(v)$ are disjoint.
Roughly, we construct split graphs $G$ s.t. conversely, if $\diam(G) =3$, then $H_2$ is an isometric subgraph.

\smallskip
\noindent
{\bf Central vertices in $1$-hyperbolic graphs.} (Theorem~\ref{thm:hardness:1-hyp}, see Sec.~\ref{sec:hardness:one-hyp} below).
The {\em Hitting Set Conjecture} posits that for any $\varepsilon > 0$, there is no algorithm that for two lists $A, B$ of $n$ subsets, can decide in $\cO(n^{2-\varepsilon})$ time if there is a set in the first list that intersects every set in the second list~\cite{AVW16}.
The authors of~\cite{AVW16} introduced the so-called HS-graphs, which they used in order to prove that distinguishing graphs with radius at most two from those of radius at least three cannot be done in truly subquadratic time, under the Hitting Set conjecture.
Furthermore, it was observed in~\cite{Chepoi2018FastEA} that every HS-graph is $2$-hyperbolic.
We refine the hyperbolicity bound on HS-graphs, using two pre-processing rules, to prove the final result of the paper.

\subsection{Recognition of $\frac 1 2$-hyperbolic graphs}\label{sec:hardness:half-hyp}

\begin{mylemma}[\cite{OV}]\label{lem:seth-ov}
    Under SETH, for any $\varepsilon > 0$, there exists a constant $\gamma > 0$ such that we cannot solve {\sc OV} in $\cO(n^{2-\varepsilon})$ time, even if $|U| \le \gamma \cdot \log{n}$.    
\end{mylemma}

Recall that a graph $G$ is called a {\em split graph} if its vertex-set can be bipartitioned in a clique and an independent set.
Split graphs are chordal graphs.
Furthermore, the hyperbolicity of chordal graphs was characterized in~\cite{hypChordal}, as follows:

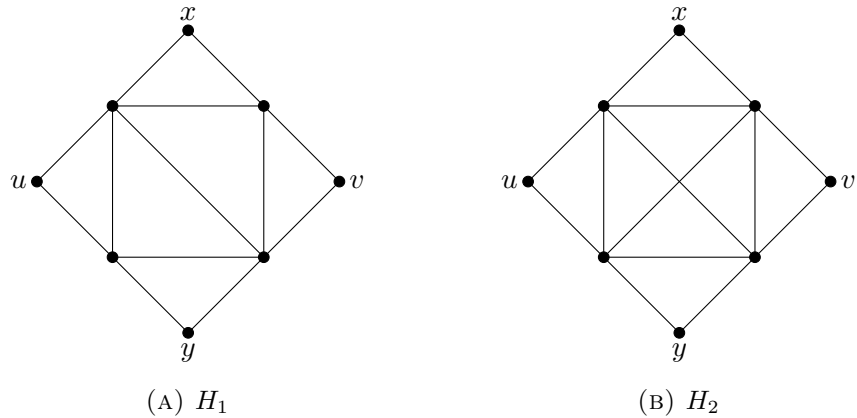
\begin{figure}[!h]
    \centering
    \begin{subfigure}[b]{.4\textwidth}
        \centering
        \begin{tikzpicture}
            \filldraw (0,2) circle (2pt) node[anchor=south]{$x$};
            \filldraw (-1,1) circle (2pt) node{};
            \filldraw (1,1) circle (2pt) node{};
            \filldraw (-1,-1) circle (2pt) node{};
            \filldraw (1,-1) circle (2pt) node{};
            \filldraw (0,-2) circle (2pt) node[anchor=north]{$y$};
            \filldraw (-2,0) circle (2pt) node[anchor=east]{$u$};
            \filldraw (2,0) circle (2pt) node[anchor=west]{$v$};
            \draw (0,2) -- (-1,1) -- (-2,0) -- (-1,-1) -- (0,-2) -- (1,-1) -- (2,0) -- (1,1) -- (0,2);
            \draw (-1,1) -- (1,1) -- (1,-1) -- (-1,-1) -- (-1,1) -- (1,-1);
        \end{tikzpicture}
        \caption{$H_1$}
        \label{fig:h1}
    \end{subfigure}
    \begin{subfigure}[b]{.4\textwidth}
        \centering
        \begin{tikzpicture}
            \filldraw (0,2) circle (2pt) node[anchor=south]{$x$};
            \filldraw (-1,1) circle (2pt) node{};
            \filldraw (1,1) circle (2pt) node{};
            \filldraw (-1,-1) circle (2pt) node{};
            \filldraw (1,-1) circle (2pt) node{};
            \filldraw (0,-2) circle (2pt) node[anchor=north]{$y$};
            \filldraw (-2,0) circle (2pt) node[anchor=east]{$u$};
            \filldraw (2,0) circle (2pt) node[anchor=west]{$v$};
            \draw (0,2) -- (-1,1) -- (-2,0) -- (-1,-1) -- (0,-2) -- (1,-1) -- (2,0) -- (1,1) -- (0,2);
            \draw (-1,1) -- (1,1) -- (1,-1) -- (-1,-1) -- (-1,1) -- (1,-1);
            \draw (1,1) -- (-1,-1);
        \end{tikzpicture}
        \caption{$H_2$}
        \label{fig:h2}
    \end{subfigure}
    \caption{Some $1$-hyperbolic chordal graphs~\cite{hypChordal}.}
    \label{fig:chordal-hyp}
\end{figure}
\begin{mylemma}[\cite{hypChordal}]\label{lem:hyp-chordal}
    A chordal graph is $\frac 1 2$-hyperbolic if and only if it does not contain any of $H_1,H_2$ in Fig.~\ref{fig:chordal-hyp} as an isometric subgraph (otherwise, it is $1$-hyperbolic).    
\end{mylemma}

Since $H_1$ is not a split graph, the characterization can be simplified for split graphs:
\begin{mycorollary}\label{cor:hyp-split}
    A split graph is $\frac 1 2$-hyperbolic if and only if it does not contain $H_2$ in Fig.~\ref{fig:chordal-hyp} as an isometric subgraph (otherwise, it is $1$-hyperbolic).    
\end{mycorollary}

\begin{proof}[Proof of Theorem~\ref{thm:hardness:half-hyp}]
        Let $\varepsilon \in (0,1)$ be arbitrary.
        In what follows, we prove that under SETH we cannot recognize the $\frac 1 2$-hyperbolic split graphs with $n$ vertices and at most $n^{1+o(1)}$ edges in $\cO(n^{2-\varepsilon})$ time.
        For that, let $\langle A,B,U \rangle$ be an {\sc OV}-instance, where $A$ and $B$ are families of $n$ sets over a common universe $U$.
        By Lemma~\ref{lem:seth-ov}, we can assume that $|U| \le \gamma \cdot \log{n}$, for some constant $\gamma$ that only depends on $\varepsilon$.
        Let $A_0,A_1$ be disjoint copies of $A$; for every $a \in A$, let $a_0 \in A_0$ and $a_1 \in A_1$ be denoting its copies.
        In the same way, let $B_0,B_1$ be disjoint copies of $B$; for every $b \in B$, let $b_0 \in B_0$ and $b_1 \in B_1$ be denoting its copies.
        Let also $w_0,w_1,z_A,z_B$ be elements that are not in $U$.

        The graph $G=(V,E)$ is constructed as follows:
        \begin{itemize}
            \item $V = A_0 \cup A_1 \cup B_0 \cup B_1 \cup U \cup \{w_0,w_1,z_A,z_B\}$;
            \item $A_0 \cup A_1 \cup B_0 \cup B_1$ is an independent set, and $U \cup \{w_0,w_1,z_A,z_B\}$ is a clique;
            \item for every $a \in A$, for every $u \in U$ such that $u \in a$, we add the two edges $ua_0,ua_1$;
            \item for every $b \in B$, for every $u \in U$ such that $u \in b$, we add the two edges $ub_0,ub_1$;
            \item for every $i \in \{0,1\}$, we add an edge between $w_i$ and every vertex of $A_i \cup B_{1-i}$;
            \item finally, we add an edge between $z_A$ and every vertex of $A_0 \cup A_1$, and we add an edge between $z_B$ and every vertex of $B_0 \cup B_1$.
        \end{itemize}
        By construction, $G$ is a split graph.
        Furthermore, we can compute $G$ from $\langle A,B,U \rangle$ in $\cO(n|U|) = \cO(n\log{n})$ time.
        To complete the reduction, we prove in what follows $G$ is $\frac 1 2$-hyperbolic if and only if for every $a \in A$, and for every $b \in B$, $a \cap b \ne \emptyset$. 
        In one direction, let $G$ be $\frac 1 2$-hyperbolic.
        Suppose for the sake of contradiction there exist $a \in A, b \in B$ such that $a \cap b = \emptyset$.
        Then, $a_0,a_1,b_0,b_1,w_0,w_1,z_A,z_B$ induces a copy of $H_2$.
        Furthermore, $N(a_0) = a \cup \{w_0,z_A\}$, $N(b_0) = b \cup \{w_1,z_B\}$, and so, $d(a_0,b_0) = 3$.
        In the same way, $d(a_1,b_1) = 3$.
        As a result, $H_2$ is an isometric subgraph of $G$, thus contradicting Corollary~\ref{cor:hyp-split}.
        In the other direction, let $G$ be having hyperbolicity at least one.
        By Corollary~\ref{cor:hyp-split}, $H_2$ is an isometric subgraph of $G$.
        In particular, there exist vertices $u$ and $v$ of $G$ such that $d(u,v) = 3$.
        By construction, the subsets $A_0 \cup A_1, B_0 \cup B_1, A_0 \cup B_1 \ \text{and} \ A_1 \cup B_0$ have diameter two in $G$.
        Therefore, there is some $i \in \{0,1\}$ such that exactly one of $u$ or $v$ is contained in $A_i$, and the other is contained in $B_i$.
        Up to symmetries, let $u \in A_0, v \in B_0$.
        Let $a \in A, b \in B$ be chosen such that $u = a_0, v = b_0$.
        Then, $N(u) = a \cup \{w_0,z_A\}$, and $N(v) = b \cup \{w_1,z_B\}$.
        Since $d(u,v) = 3$, we get $N(u) \cap N(v) = \emptyset$.
        Therefore, $a \cap b = \emptyset$.
\end{proof}

        \subsection{Central vertices in $1$-hyperbolic graphs}\label{sec:hardness:one-hyp}

We need the following more precise version of the Hitting Set Conjecture: there is no $\varepsilon > 0$ such that for all $\gamma \ge 1$, there is an algorithm that given two lists $A, B$ of $n$ subsets of a universe $U$ of size at most $\gamma \log{n}$, can decide in $\cO(n^{2-\varepsilon})$ time if there is a set in the first list that intersects every set in the second list~\cite{AVW16}.

\begin{proof}[Proof of Theorem~\ref{thm:center:half-hyp}]
    Let $\varepsilon \in (0,1)$ be arbitrary.
    In what follows, we prove that under the Hitting Set conjecture, we cannot solve the {\sc Center} problem in $\cO(n^{2-\varepsilon})$ time on $1$-hyperbolic graphs with $n$ vertices and at most $n^{1+o(1)}$ edges.
    For that, let $\langle A,B,U \rangle$ be an HS-instance such that $A$ and $B$ are families of $n$ sets over a common universe $U$, such that $|U| \le \gamma \cdot \log{n}$ for some constant $\gamma$ only depending on $\varepsilon$.

    The HS-graph $G=(V,E)$ is defined in~\cite{AVW16} as follows:
    \begin{itemize}
        \item $V = A \cup B \cup U \cup \{x,y,z\}$;
        \item for every $a \in A$, for every $u \in U$, $au \in E$ if and only if $u \in a$;
        \item for every $b \in B$, for every $u \in U$, $bu \in E$ if and only if $u \in b$;
        \item we add an edge between $x$ and every vertex of $A \cup U$;
        \item we add an edge between $y$ and every vertex of $A$;
        \item finally, we add the edge $yz$.
    \end{itemize}
    The graph $G$ can be constructed from $\langle A,B,U \rangle$ in $\cO(n|U|) = \cO(n\log{n})$ time.
    Furthermore, it was proved in~\cite{AVW16} $\rad(G) = 2$ if and only if there exists a set $a \in A$ such that, for every $b \in B$, $a \cap b \ne \emptyset$.
    By construction, $\diam(G) \ge 4$; graphs of diameter at most four are $2$-hyperbolic.
    In what follows, we pre-process the HS-instance $\langle A,B,U \rangle$ so that the resulting HS-graph is $1$-hyperbolic.
    Since $z$ is a pendant vertex, and the hyperbolicity of a graph is the maximum hyperbolicity over its biconnected components~\cite{CliqueSep}, we first observe that both the graphs $G$ and $G \setminus z$ have the same hyperbolicity value.
    Furthermore, we recall that every graph of diameter at most three is $1$-hyperbolic~\cite{KoMo}.
    Therefore, to prove that $G$ is $1$-hyperbolic, it suffices to prove $\diam(G \setminus z) \le 3$.
    For that, we pre-process the HS-instance $\langle A,B,U \rangle$ as follows:
    \begin{enumerate}
        \item Let $U_A = \{ u \in U : \exists a \in A \ \text{s.t.} \ u \in a \}$.
        Similarly, let $U_B = \{u \in U : \exists b \in B \ \text{s.t.} \ u \in b\}$.
        We replace $\langle A,B,U\rangle$ with a new HS-instance $\langle A',B',U'\rangle$, where $A' = \{ a \cap U_B : a \in A \}$, $B' = \{ b \cap U_A : b \in B \}$, and $U' = U_A \cap U_B$.
        This first transformation does not affect the existence of a hitting set in $A$.
        \item Let $u_B \notin U$ be a fresh new element. 
        We replace $\langle A,B,U\rangle$ with a new HS-instance $\langle A,B',U\cup\{u_B\}\rangle$, where $B' = \{b \cup \{u_B\} : b \in B\}$.
        In doing so, the sets of $B$ can be assumed to pairwise intersect.
        Furthermore, this second transformation does not affect the existence of a hitting set in $A$ either.
    \end{enumerate}
    Apply both transformations above sequentially.
    Then, let $G$ be the resulting HS-graph.
    Since $B_1(x) = A \cup U \cup \{x\}$, the vertices of this ball are pairwise at distance at most two.
    Furthermore, $A \subseteq N(y)$, and $B \subseteq N(U)$.
    Therefore, every vertex of $B \cup \{y\}$ is at a distance at most three to every vertex of $B_1(x)$.
    By the second transformation above, the vertices of $B$ are pairwise at distance two.
    Therefore, to prove that $\diam(G \setminus z) \le 3$, we are left proving that $y$ is at a distance at most three to every vertex of $B$.
    By the first transformation above, every vertex of $U \setminus \{u_B\}$ has a neighbour both in $A$ and $B$.
    In particular, every vertex of $B$ must be at distance two to some vertex in $A \subseteq N(y)$.
    The latter indeed proves that $d(y,b) \le 3$ for every $b \in B$.
\end{proof}

Let $u_0,u_1,u_2$ be vertices in a graph $G$.
We can define a {\em geodesic triangle} $\Delta(u_0,u_1,u_2)$ as the union of three shortest-paths $P_0,P_1,P_2$, where for each $i$ $P_i$ is a shortest $u_iu_{i+1}$-path (indices are taken modulo $3$).
Furthermore, $P_0,P_1,P_2$ are called the sides of the triangle.
The geodesic triangle is called {\em $k$-slim} if for each $i$, for every vertex $v_i \in V(P_i)$, $d\left(v_i,P_{i+1}\cup P_{i+2}\right) \le k$.
Finally, the least $k$ such that every geodesic triangle is $k$-slim is called the {\em slimness} of $G$.
By definition, the slimness of a graph must be a natural number.
Graphs with slimness $0$ are exactly the block graphs~\cite{Slimness}.
However, we can prove that our construction for Theorem~\ref{thm:hardness:1-hyp} outputs some HS-graphs of slimness at most $1$.
For that, we use the following observations:
\begin{itemize}
    \item If $z$ is a pendant vertex of a graph $G$, then the slimness constants for $G$ and $G \setminus z$ are the same.
    Indeed, since $G \setminus z$ is an isometric subgraph of $G$, its slimness is at most that of $G$.
    Furthermore, every geodesic triangle $\Delta(u_0,u_1,u_2)$ of $G$ that is not a geodesic triangle of $G \setminus z$ satisfies $z \in \{u_0,u_1,u_2\}$.
    By symmetry, let $u_0 = z$.
    Then, both shortest paths $P_0,P_2$ contain the extremal edge $zy$, where $y$ is the unique neighbor of $z$ in $G$.
    It implies that $\Delta(u_0,u_1,u_2)$ is the union of the edge $yz$ with a geodesic triangle $\Delta'(y,u_1,u_2)$ of $G \setminus z$.
    Hence, the slimness of $G$ is at most that of $G \setminus z$.
    \item The slimness of a graph $G$ is at most $\left\lfloor \frac{\diam(G)}2\right\rfloor$.
    Indeed, this is because for every geodesic triangle $\Delta(u_0,u_1,u_2) = P_0 \cup P_1 \cup P_2$, for each $i$, for every vertex $v_i \in V(P_i)$,
    $$d\left(v_i,P_{i+1}\cup P_{i+2}\right) \le \min\{d(v_i,u_i),d(v_i,u_{i+1})\} \le \left\lfloor \frac{d(u_i,u_{i+1})}2 \right\rfloor \le \left\lfloor \frac{\diam(G)}2\right\rfloor$$
\end{itemize}
Overall, under the Hitting Set conjecture, the {\sc Center} problem requires essentially quadratic time on graphs with slimness at least one.


\bibliographystyle{amsplain}
\bibliography{biblio.bib}

\end{document}